\documentclass[letterpaper,12pt]{article}
\usepackage{amsfonts}
\usepackage[pdftex]{graphicx}
\usepackage{amsmath,amssymb,amsthm}
\usepackage{natbib}
\usepackage{hyperref}
\usepackage{bm}
\usepackage{fancyhdr}
\usepackage{sectsty}
\usepackage[notcite,notref,final]{showkeys}
\usepackage{paralist} 
\usepackage{tikz}
\usepackage{stmaryrd}
\usepackage{setspace}
\usepackage{todonotes}
\usepackage{multirow}
\usepackage{comment}

\usepackage[left=1.2in, right=1.2in, top=1in, bottom=1in]{geometry}

\definecolor{light-gray}{gray}{0.8}

\newtheoremstyle{myplain}
  {9pt}
  {9pt}
  {\itshape}
  {\parindent}
  {\scshape}
  {:}
  {.4em}
  {}
\newtheoremstyle{mydefinition}
  {9pt}
  {9pt}
  {\itshape}
  {\parindent}
  {\scshape}
  {:}
  {.4em}
  {}
\newtheoremstyle{myremark}
  {9pt}
  {9pt}
  {}
  {\parindent}
  {\scshape}
  {:}
  {.4em}
  {}
\theoremstyle{myplain}
\newtheorem{theorem}{Theorem}
\newtheorem{corollary}{Corollary}
\newtheorem{lemma}{Lemma}
\newtheorem{proposition}{Proposition}
\theoremstyle{mydefinition}
\newtheorem{assumption}{Assumption}
\newtheorem{definition}{Definition}
\theoremstyle{myremark}
\newtheorem{example}{Example}
\newtheorem{remark}{Remark}
\renewcommand{\cite}{\citet}
\def\argmax{\mathop{\rm arg\,max}}

\newcommand{\cA}{\mathcal{A}}
\newcommand{\cC}{\mathcal{C}}
\newcommand{\cD}{\mathcal{D}}
\newcommand{\cF}{\mathcal{F}}

\newcommand{\cL}{\mathcal{L}}

\newcommand{\cP}{\mathcal{P}}

\newcommand{\cS}{\mathcal{S}}

\newcommand{\cX}{\mathcal{X}}
\newcommand{\cY}{\mathcal{Y}}
\newcommand{\cZ}{\mathcal{Z}}
\newcommand{\cV}{\mathcal{V}}
\newcommand{\cU}{\mathcal{U}}
\hypersetup{colorlinks=true, linkcolor=blue, citecolor=blue} 

\newcommand{\citeposs}[1]{\citeauthor{#1}'s \citeyearpar{#1}}
\DeclareMathOperator{\Sel}{Sel}
\DeclareMathOperator{\cl}{cl}
\newcommand{\contf}{\mathbb C_\theta} 
\newcommand{\capf}{\mathbb C^*_\theta} 

\newcommand{\eX}{{\boldsymbol{X}}}
\newcommand{\eY}{{\boldsymbol{Y}}}   
\newcommand{\eV}{{\boldsymbol{V}}}

\newcommand{\eU}{{\boldsymbol{U}}}
\newcommand{\eW}{{\boldsymbol{W}}}
\newcommand{\eK}{{\boldsymbol{K}}}

\newcommand{\sfo}{\mu} 
\newcommand{\sfs}{\pi}   

\numberwithin{equation}{section}

\title{Set-Valued Control Functions\thanks{We appreciate feedback from Marinho Bertanha, Andrew Chesher, Ben Deaner, \'{A}ureo de Paula, Dongwoo Kim, Pat Kline, Francesca Molinari, Kenichi Nagasawa, Adam Rosen, Andrei Zeleneev, and participants in seminars at Johns Hopkins, Stanford, Berkeley, Yale, UCSC, UC Davis, UCL, and Warwick and workshops at Manchester and Seoul National, as well as at the North American Winter Meeting 2025, the North American Summer Meeting 2024, and the Asian Meeting 2024 of the Econometric Society. We thank Shuowen Chen, Qifan Han, and Zhanyuan Tian for their excellent research assistance. We gratefully acknowledge financial support from NSF grant SES-2018498.}}

\author{
Sukjin Han \\  School of Economics \\ University of Bristol \\ \url{vincent.han@bristol.ac.uk} 
\and
Hiroaki Kaido \\ Department of Economics \\ Boston University \\ \url{hkaido@bu.edu}  
}

\begin{document}
\onehalfspacing

\maketitle

\begin{abstract}
The \emph{control function} approach allows the researcher to identify various causal effects of interest. While powerful, it requires a strong invertibility assumption in the selection process, which limits its applicability. This paper expands the scope of the nonparametric control function approach by allowing the control function to be \emph{set-valued} and derive sharp bounds on structural parameters. The proposed generalization accommodates a wide range of selection processes involving discrete endogenous variables, random coefficients, treatment selections with interference, and dynamic treatment selections. The framework also applies to partially observed or identified controls that are directly motivated from economic models.

\vspace{0.3in}
\end{abstract}

\clearpage

\section{Introduction}
Endogeneity is the main challenge in conducting causal inference with
observational data. The control function (CF) approach has been a
valuable tool in addressing endogeneity and recovering various causal
parameters. Although this approach originated in parametric models (e.g., \citealp{dhrymes1970econometrics, heckman1979sample}), it
has been proven to be a powerful tool for identification and estimation
in nonparametric models that allow causal effect heterogeneity. The
CF approach constructs control variables $V$, which define a \emph{latent
type} conditional on which endogenous explanatory variables $D$ can
be viewed as unconfounded. In observational settings, such $V$ is
typically constructed by inverting treatment selection processes so
that it is written as a function of observables---thus a \emph{control
function}. Many empirical studies build on this insight to construct
and utilize control variables \citep{olley1996dynamics,levinsohn2003estimating,ackerberg2015identification,Kline:2016aa,Card:2019aa,Abdulkadiroglu:2020aa,Bishop:2022aa}. While powerful, this approach
relies on the \emph{invertibility} of selection models. For example,
in nonparametric triangular models, invertibility requires $D$ to
be continuously distributed and the selection equation for $D$ to
be strictly monotone in a scalar unobservable variable. This type
of restriction is viewed as the most important limitation of the
CF approach \citep{Blundell:2003wi}. More generally, whether selection models are involved or not, empirical researchers encounter situations where control variables are only partially observed or identified; see below.

This paper allows the control function to be set-valued. Formally, a \emph{set-valued control function} $\eV$ is a random
closed set, constructed from observable variables, that contains the
true control variable $V$ inside it. Based on this insight, we build a general framework of sharp-bound characterization and inference that expands the scope of the CF approach to a variety of contexts. When selection processes are involved, this approach allows us to drop invertibility. This adaptation accommodates a wide range
of selection processes. Observational data are often generated through
complex selection processes. For example, $D$ can be a binary variable
generated by a generalized Roy model \citep{Eisenhauer:2015aa}. One can allow for richer heterogeneity by considering a selection
model with binary $D$ that violates the local average treatment effect
(LATE) monotonicity \citep{Imbens:1994tc} or, analogously, a model
with continuous $D$ with vector unobservables (e.g., selection with
random coefficients). Other examples are the cases where $D$ is determined
through interaction of multiple agents \citep{Tamer:2003tr,ciliberto2021market,Balat:2022aa} (e.g., due to violation of
the stable-unit treatment value assumption (SUTVA) in forming outcomes);
where $D$ and outcomes are dynamically determined over time \citep{Han:2021aa,han2023optimal}; and where $D$ results from censoring \citep{Manski:2002um} or as corner solutions. Such processes typically
violate the invertibility assumption, as the mapping from observables
to $V$ is only a correspondence. We show that the CF approach can
still be used with these selection processes to partially identify
structural (i.e., causal) parameters, such as average and quantile structural functions
for outcomes. By allowing control functions to be set-valued, we can also incorporate a wider range of examples that use controls without relying on selection models. \cite{bertanha2024causal} control for a set of true preferences using strategic reports, while \cite{auerbach2022identification} recovers a control variable from a friendship network. In other examples, controls are simply interval data (e.g., wealth, debt, biometric measures, psychological traits). Our framework can be applied to such scenarios, enabling researchers to conduct sensitivity analyses.

To our knowledge, the general form of sharp identifying restrictions under the control function assumption were unknown without requiring the full observability of the control variable. 
Our innovation is to construct a random set that contains outcome values consistent with the control function assumption by combining a set-valued control function $\eV$ with an augmented outcome equation. This is a crucial step to formulate the model's \emph{incomplete prediction},  which is the main contribution of this paper. This step enables us to utilize tools from the theory of random sets \citep{Molchanov:2017th}, such as the containment functional and Aumann expectation of the set-valued prediction, thereby establishing restrictions that lead to the sharp identified set for the structural parameters. Assuming full independence of treatments conditional on controls (Assumption \ref{as:cf_dist}), the identifying restrictions result in inequality constraints on the conditional choice probabilities. If we assume a weaker conditional mean independence (Assumption \ref{as:cf_mean}), these restrictions become conditional moment inequality restrictions. Inference methods based on such restrictions have been extensively studied \citep[see e.g.][]{Canay:2017aa,Molinari:2020un}. Hence, practitioners can apply existing methods to our identifying restrictions. To demonstrate this point, we illustrate our approach by studying individuals' HIV preventive behavior under an informational provision intervention studied by \cite{Thornton:2008tp}. This empirical illustration involves an ordered outcome, an endogenous treatment, multiple instruments, and various observed controls. We examine the effectiveness of the intervention by constructing CIs for policy-relevant parameters such as the counterfactual switching probability of individual choices.

This paper contributes to the vast literature on identification and
estimation in nonparametric models with endogenous explanatory variables. In linear models, the two-stage least
squared (TSLS) estimator can have two different interpretations: the
instrumental variable (IV) approach and the CF approach \citep{Blundell:2003wi}. The noparametric version of the IV approach is considered
in \cite{ai2003efficient, newey2003instrumental, hall2005nonparametric, chernozhukov2005iv,blundell2007semi, chen2009efficient, darolles2011nonparametric, chen2012estimation, dhaultfoeuille2015identification, torgovitsky2015identification, vuong2017counterfactual, chen2018optimal}.
Typically, this approach assumes invertibility in the outcome equation
and thus relies on a scalar unobservable, so that the IV assumption
can be utilized. The CF approach is generalized to nonparametric models by \cite{Newey:1999tu, chesher2003identification, das2003nonparametric, Blundell:2004td, Imbens:2009vs,DHaultfoeuille:2021aa,Newey:2021ab,nagasawa2024treatment}, following the adaptation to nonlinear parametric models in \cite{Newey:1987aa, RIVERS1988347, smith1986exogeneity, blundell1989estimation}. The nonparametric CF literature typically assumes a model for endogenous
explanatory variables and its invertibility in a scalar unobservable. Then, this approach generates control variables and combines it with
the CF assumption (non-nested to the IV assumption) to identify structural parameters. Although this
approach restricts selection behavior and is not applicable to discrete
treatments, its advantage is the freedom from restricting heterogeneity
directly relevant in generating causal effects. Another important
strand of the causal inference literature concerns a binary or discrete
treatment with a monotonicity assumption \citep{Imbens:1994tc, abadie2002instrumental} or equivalently
\citep{Vytlacil:2002vo} a threshold-crossing model that involves a scalar unobservable \citep{Heckman:2005aa}.

\cite{Chesher:2017vu} propose a generalized IV (GIV) framework that allows for partial identification of structural parameters in a range of complete and incomplete models (see also \cite{Chesher:2012vn, Chesher:2013vo}) under stochastic restrictions governing the relationship between the latent variables and instrumental variables.
 \cite{Chesher:2017vu} applies their analysis to single-equation IV models, extending the IV approach to partial identification. This paper proposes the CF approach to partial identification, filling the gap in the literature. Sharing the aspect of the CF literature above, we allow for arbitrary causal effect heterogeneity (e.g., multi-dimensional outcome unobservables). Overcoming the aspect of the CF literature, we accommodate discrete treatments along with heterogeneity and complexity in treatment selection (e.g., multi-dimensional selection unobservables, incompleteness of selection models).

An incomplete model can often be formulated using one of the following structures:
(a) Given covariates $X\in\mathcal X$, latent variables $U\in\mathcal U$, and parameters $\theta\in\Theta$, the model predicts a set $\eY(X,U;\theta)\subseteq \mathcal Y$ of values for the outcome  $Y\in\mathcal Y$; or (b) Given covariates $X\in\mathcal X$ and observed outcome $Y\in \mathcal Y$ and parameters $\theta\in\Theta$, the model predicts a set $\eU(X,Y;\theta)\subseteq \mathcal U$ for the latent variable  $U\in\mathcal U.$ 
One can choose how to formulate an incomplete model depending on the objects of interest, assumptions, and observability of variables. For example, many discrete choice models that specify a parametric family for $F$ can be formulated using (a) \citep{jovanovic1989observable,galichon2011set}. \cite{beresteanu2011sharp} study an extended setting, allowing for solution concepts that involve randomization, e.g., correlated equilibria.
\cite{Chesher:2017vu} employ (b) to define a set of latent variables whose selection satisfies specific stochastic restrictions. A similar approach is used in more recent work by \cite{chesher2023iv} for IV Tobit models and \cite{chesher2023identification} in the context of panel data models. 
Our approach to defining a random set closely relates to the previous work but differs in the following respects. The first step of our approach uses (b) to construct the set-valued CF as a set of unobserved control variables. The second step plugs the set-valued control function into an augmented outcome equation to define a model in structure (a). This hybrid approach allows us to extend the control function approach naturally to the current setting.

\cite{Chesher:2005tu} also considers partial identification without requiring invertibility in selection processes. He assumes that discrete endogenous variables are generated from ordered structure and focuses on local parameters. This paper in contrast focuses on global parameters, while encompassing a range of selection processes including ordered selection. Recently in a sample selection model, \cite{aradillas2024inference} considers control variables that are partially observed either as interval data or via economic theory and constructs a confidence set for partially identified parameters. The framework of the current paper is more general and focuses on a different source of set-valued CF (i.e., non-invertibility). \cite{Shaikh:2011vk, Jun:2011we, mourifie2015sharp, Mogstad:2018aa, machado2019instrumental, han2024computational} consider
partial identification in nonparametric models without requiring invertibility in selection processes; they consider either a binary treatment generated from threshold-crossing models (equivalently, under the LATE monotonicity) or a discrete treatment with similar restrictions. These models are nested within the class of models we consider, but our distinct features include the generality in selection processes and the use of the CF approach.

\section{Setup}
Let $Y\in \cY\subseteq\mathbb R^{d_Y}$ be the outcome of interest generated according to the following \emph{outcome equation}:
\begin{align}
Y=\sfo(D,X,U),	\label{eq:outcome_eq}
\end{align}
where $D\in\cD\subseteq\mathbb R^{d_D}$ is a vector of endogenous treatment variables, $X\in \cX\subseteq\mathbb R^{d_X}$ is a vector of covariates, and $U\in \mathcal U\subseteq\mathbb R^{d_U}$ is a vector of latent variables. All random variables are defined on a complete probability space $(\Omega,\mathfrak F, P)$.
The \emph{structural function} $\sfo$ determines the value of the \emph{potential outcome} $Y(d)=\sfo(d,X,U)$ that would realize when the endogenous variable is set to $ d\in \cD$.
Many policy-relevant parameters are features of the potential outcome, and hence functionals of $\sfo$. Examples are the \emph{average structural function} and the \emph{distributional structural function}: $\text{ASF}(d)\equiv E[\sfo(d,X,U)]=E[Y(d)]$ and $\text{DSF}(d)\equiv F_{\sfo(d,X,U)}=F_{Y(d)}$, respectively. Other examples are the \emph{policy-relevant structural function} and the \emph{mediated structural function}, defined later.

A vector of control variables $V\in\mathcal{V}$ ($\subseteq\mathbb R^{d_V}$ for example) is such that, the assignment of $D$ becomes independent of $U$ once we condition on $V$ and the observable covariates $X$:
\begin{align}
    D\perp U|X,V.
\end{align}
 Such variables allow the researcher to identify various causal parameters without additional parametric assumptions on $\mu$ or the distribution of unobservables.  

For this approach to work, one needs to express $V$ as a function of observable variables. Suppose $D$ is generated from a selection process and $Z$ are the vector of instrumental variables. A commonly used specification for the selection process is the additive model $D=\Pi(Z)+V$, in which one may express $V=D-\Pi(Z)$  \citep[e.g.,][]{Newey:1999tu}.
\cite{Imbens:2009vs} consider a nonseparable system, in which a single endogenous variable is modeled as $D=h(Z,\tilde V)$ of a vector of instrumental variables $Z$ and a continuously distributed \emph{scalar} latent variable $\tilde V$,
where $h$ is strictly monotonic in $\tilde V$.  They show that, under the independence of $(U,\tilde V)$ and $Z$, one may use the conditional cumulative distribution function $V\equiv F_{D|Z}(D|Z)$ as a control variable.
The key assumption is the invertibility of $h$ in the latent variable, which ensures that there is a one-to-one relationship between $\tilde V$ and $V$.\footnote{This approach can be extended to a class of simultaneous equations models that meet a certain separability condition \citep{Blundell:2013tg,Blundell:2014wo}.} 

When $D$ is binary, there are other approaches employed in the literature to use control functions without invertibility. These approaches maintain a scalar unobservable $V$ in the selection and additive separability in the outcome equation and either (i) impose parametric assumptions, such as a parametric distribution of unobservables \citep{heckman1979sample, dal2021information}; or (ii) introduce the marginal treatment effects (MTE) \citep{Heckman:2005aa} as a control function, where a parametric restriction is imposed on the MTE function to deal with discrete ``forcing'' variables \citep{brinch2017beyond}; see \cite{Kline:2019wa} for related discussions.

The previous approaches rely on either invertibility in the selection or certain parametric restrictions. The parametric restrictions are subject to misspecification and they are often combined with scalar $V$. The invertibility requirement restricts the form of the selection equation and the dimension of $V$. When $V$ is continuously distributed, the invertibility also requires $D$ to be continuous, which limits the scope of the control function assumption. Moreover, having vector $V$ is important in allowing for rich heterogeneity in the selection process. For example, a multidimensional $V$ can capture individuals' heterogeneous responses to a determinant of the treatment take-up decision by allowing for both compliers and defiers. We, therefore, aim to remove these restrictions. 

Another line of research builds control variables $V$  from observables that are not directly related to selection processes. For instance, \cite{auerbach2022identification} utilizes friendship networks to infer a control for latent student abilities. Our framework can also be employed to construct a set-valued control that relies on less stringent assumptions, thereby enabling the researcher to perform a sensitivity analysis.

\subsection{Motivating Examples}\label{ssec:examples}

A \emph{set-valued control function} $\eV$ is a random closed set that contains the true control $V$ almost surely and is a function of observable variables. We state this as a formal assumption in the next section (Assumption \ref{as:cf_set}), together with a precise measurability requirement.
 
Before proceeding, we introduce motivating examples. The examples share the following features. First, they involve control variables, conditional on which the treatment decisions can be viewed as random. Second, they do not allow the researcher to uniquely recover the control variables. Nonetheless, it is possible to construct a set-valued control function. Finally, the above features are related to the fact that the control variable $V$ may be interpreted as structural unobservables in these examples.

We start with examples involving selection processes (Examples \ref{ex:Roy1}-\ref{ex:vector_D}). When the control variable appears as a component in a selection process, the above features can be summarized in the following generalized selection equation:
\begin{align}
D=\pi(Z,X,V).	\label{eq:gen_sel_eq}
\end{align}
Note that $(D,X,Z)\mapsto V$ is, in general, a correspondence because
either $\pi(Z,X,\cdot)$ is not necessarily strictly monotonic or $V$ is not scalar. Therefore, the selection process \eqref{eq:gen_sel_eq} only restricts $V$ to the following set almost surely: $\{v:D=\pi(Z,X,v)\}\subseteq\mathbb{R}^{d_{V}}$. Motivated by this, we define the set-valued control function as follows in the following two examples:
\begin{align}
\eV(D,Z,X;\pi)\equiv\cl\{v:D=\pi(Z,X,v)\}\subseteq\mathbb{R}^{d_{V}}.	\label{eq:gen_eV}
\end{align}
We define $\eV$ as a closed set in order to utilize the theory of random sets.
In Examples \ref{ex:Roy1}-\ref{ex:vector_D}, we illustrate specific forms of \eqref{eq:gen_sel_eq} and \eqref{eq:gen_eV}.

\begin{example}[Binary and Censored Treatment Decisions]\label{ex:Roy1}
Let $D$ be a binary treatment that is determined by the selection equation
\begin{align}
	D=1\{\sfs(Z,X)\ge V\},\label{eq:roy_sel2}
\end{align}
where we normalize $V|X$ to the uniform distribution without loss of generality. The selection equation can be motivated by the generalized Roy model \citep{Eisenhauer:2015aa}. Suppose $Y=DY(1) + (1-D)Y(0)$ where $Y(d)$ follows
\begin{align}
Y(d)& =\mu(d,X)+U_{d}\quad\text{for }d=0,1.
\end{align}
We allow the unobservables $U_d$ to be treatment-specific. This makes $Y$ a function of a vector unobservable, $U\equiv(U_1,U_0)$. The treatment decision is based on the net surplus $S$ from the treatment:
\begin{align}
D  =1\{S\ge0\}\equiv1\{Y(1)-Y(0)-C\ge0\},\label{eq:roy_sel1}
\end{align}
where $C\equiv\mu_{c}(Z,X)+U_{c}$ be the cost of choosing one alternative over the other, and $Z$ is a vector of variables that shifts the cost but not the outcome.\footnote{The generalized Roy model above nests the classical Roy model where
$C$ is degenerate \citep{Heckman:1990aa} and the extended Roy model where $U_{c}$ is
degenerate \citep{Heckman:2007aa}.} We may write the surplus as $S =\sfs(Z,X)-V$, where $\sfs(Z,X)\equiv\mu(1,X)-\mu(0,X)-\mu_{c}(Z,X)$ is the observable part of the surplus, and   $V\equiv U_{c}-U_{1}+U_{0}$ is the unobserved part of the surplus. Then, we can express the treatment decision as \eqref{eq:roy_sel2}. Clearly, $V$ depends on $(U_0,U_1)$. 

Suppose we are interested in the causal effect of $D$ on $Y$. Suppose $Z$ is independent of $U$ given $(X,V)$.
Then, $(X,V)$ are valid control variables because $D$'s remaining variation is independent of $U$ conditional on them. What prevents us from applying the existing approach is that we cannot recover $V$ by inverting \eqref{eq:roy_sel2} because $D$ is binary. Nonetheless, the model restricts $V$ to the following set almost surely:
\begin{align}
\eV(D,Z,X;\sfs) & =\begin{cases}
[0,\sfs(Z,X)] & \text{if }D=1\\{}
[\sfs(Z,X),1] & \text{if }D=0,
\end{cases}\label{eq:Roy1_eV}
\end{align}
which is a set-valued analog of the control function we may condition on. 

The previous specification satisfies the LATE monotonicity, eliminating either compliers or defiers
\citep{Imbens:1994tc,Vytlacil:2002vo}. 
Next, we consider a selection model that allows richer compliance types, and thus increased heterogeneity in the population. Suppose the value of the instrument is set to $z$. Let the potential treatment be
\begin{align}
D(z) & =1\{\sfs(z,X)\ge V_{z}\}\quad\text{for }z\in\cZ,\label{eq:Roy2}
\end{align}
where $V_z$ is an unobservable specific to the value of $Z$. The observed treatment is $D=\sum_{z\in \cZ}D(z)1\{Z=z\}$. Suppose $Z$ is binary below. Given \eqref{eq:Roy2}, both compliers and defiers can have
nonzero shares:
\begin{align*}
\{D(0)=0,D(1)=1\} & =\{V_{0}>\sfs(0,X),V_{1}\le\sfs(1,X)\},\\
\{D(0)=1,D(1)=0\} & =\{V_{0}\le\sfs(0,X),V_{1}>\sfs(1,X)\}.
\end{align*}
The observed treatment $D$ is a function of $(V_0,V_1)$, satisfying
\begin{align} 
D & =1\{D(0) + (D(1)-D(0))Z\ge0\} \equiv1\{\tilde{\sfs}(Z,X)+(V_{1}-V_0)Z+V_{0}\ge0\},\label{eq:nonmono_sel}
\end{align}
where $\tilde{\sfs}(Z,X)\equiv\sfs(0,X)+Z(\sfs(1,X)-\sfs(0,X))$.\footnote{Note that $D=ZD(1) + (1-Z)D(0)=1\{Z(\pi(1,X)-V_1)+(1-Z)(\pi(0,X)-V_0)\ge0\}$.}
One may view the last expression as a \emph{random-coefficient} model, in which the individuals respond heterogeneously to interventions to $Z$ \citep{Gautier:2011va,Kline:2019wa}. Our framework allows us to proceed with mild assumptions, which may not be sufficient for point identification as in the previous work. Suppose the outcome $Y$ is generated according to \eqref{eq:outcome_eq} and $Z$ is independent of $U$ conditional on $(X,V_0,V_1)$. Then, $(X,V_0,V_1)$ are valid control variables.
By \eqref{eq:nonmono_sel}, $V\equiv(V_0,V_1)$ belongs to the following set almost surely:
\begin{align}
\eV(D,Z,X;\sfs)  =\begin{cases}
\left\{ (v_0,v_1):\tilde \sfs(Z,X)+(1-Z)v_{0}+Zv_{1}\ge0\right\}  & \text{if }D=1\\
\left\{ (v_0,v_1):\tilde \sfs(Z,X)+(1-Z)v_{0}+Zv_{1}\le 0\right\}  & \text{if }D=0.
\end{cases}\label{eq:Roy2_eV}
\end{align}

Similar to the binary case is $D$ being a censored decision as a corner solution in an agent's optimization problem. In this case, the generalized selection equation is not invertible and a corresponding set-valued control function can be constructed; see Section \ref{ssec:corner} in the Appendix for details. $\square$
 \end{example}

\begin{example}[Treatment Responses as Vectors]\label{ex:vector_D}
The next example involves a \emph{vector} of treatments generated by either (i) strategic decisions of multiple individuals \citep{Balat:2022aa} or (ii) a single agent's dynamic decisions over multiple periods \citep{Han:2021aa, han2023optimal, han2023semiparametric}. Let $D$ be the vector of binary decisions across individuals or the vector of binary treatments and previous outcomes over periods. We are interested in the effect of the entire profile $D$ on an outcome $Y$. To that end, suppose we have a vector of (individual- or time- specific) IVs, $Z$, and let
$\pi(\cdot)$ be the generalized selection function for $D$. Note that $\pi(\cdot)$ is not invertible in the corresponding unobservables due to the discreteness of $D$, as in Example \ref{ex:Roy1}. Moreover, in case (i), decisions across individuals can be generated from multiple equilibria and, in case (ii), decisions across time can involve dynamic endogeneity. These aspect further complicates the CF approach. However, we can construct appropriate control variables $V$ and corresponding $\eV(D,Z;\sfs)$ in multi-dimensional spaces. We detail these examples in Sections \ref{ssec:ex_dynamic4}-\ref{ssec:ex_strategic}.
\end{example}

The examples above constructed $\eV$ from selection processes. The next example concerns control variables that are not necessarily generated from selection models. Set-valued control functions $\eV(Z,X)$ in such settings are functions of variables other observables $(Z,X)$, where $Z$ is not necessarily an instrument. Often, it contains information on the true control $V$, and some components of $Z$ may be excluded from the outcome equation.

\begin{example}[Set-Valued Controls Without Selection]\label{ex:no_sel}
We provide three examples using the notation of our paper. 
First, control variables $V$ may simply be partially observed via interval or censored measurement. In many administrative data, information such as wealth, debt, biometric measures, and psychological traits is observed as an interval due to limitations in data collection or privacy concerns.\footnote{For example, wealth in the Health and Retirement Survey (HRS) and income in the Current Population Survey (CPS) are measured as intervals.} In this case, $\eV$ would be directly obtained from observed intervals and can be expressed as $\eV(Z,X)=[X_{L},X_{U}]$; see \cite{Manski:2002um} for a related setting.

In the context of school matching mechanisms, \cite{bertanha2024causal} estimate the causal effects of school assignment using students' local preferences as control variables $V\in \mathcal V$ (where $\mathcal V$ is the set of preference relations) to enable regression-discontinuity comparisons. The key feature of their setup is that, under capacity constraints, students have incentives to misreport their preferences. Based on students' reported partial order of preferences, they recover local preference sets $\eV(Z,X)$ (with reported preferences $Z$ and test scores $X$) that contain the true preference $V$ a.s., and subsequently characterize the bounds on the effects of school assignment.

In a social network setting, \cite{auerbach2022identification} considers a partial linear model with a nonparametric function $\lambda(\cdot)$ of social characteristics as an unknown control variable $V$, which is seldom identified to be used as a control function. Instead, he proposes to use the link function $\pi(\cdot)$ in a nonparametric link formation model that is identified from the distribution of social links. Then, under the assumption that individuals with similar link functions have similar values of control $\lambda(V)$ (Assumption 3 therein), he identifies the slope parameters. However, one may want to relax this identifying assumption and allow individuals with similar link functions to have values of $\lambda(V)$ with discrepancy bounded by a sensitivity margin. This can be achieved by constructing a set-valued control $\eV=\{(v,v'):\|\pi_v-\pi_{v'}\|_{L^2}\le \delta\}$ that contains the latent characteristics of a pair of individuals whose link functions are within a certain distance $\delta$.
Then, we can recover a set of controls, $\lambda(\eV)$, and partially identify the slope parameters. $\square$
\end{example}

\section{Model Prediction}
As a preparation for the identification analysis,
we formulate the model prediction. We show that our limited knowledge of the unobserved control variable $V$ can be formalized as an \emph{incomplete model}.

Throughout, $V:\Omega\to\mathcal V$ is a random element taking values in a Polish space $\cV$.
First, we assume that $V$ and observable covariates $X$ form control variables.
\begin{assumption}\label{as:cf_dist}
$U|D,X,V\sim U|X,V$.
\end{assumption}
By Assumption \ref{as:cf_dist}, the treatment decision is independent of $U$ once we condition on $(X,V)$. Next, we introduce \emph{random closed sets} and their \emph{measurable selections} \citep[see][]{Molchanov:2018ui}.

\begin{definition}[Random Closed Set]
A map $\eX$ from a probability space $(\Omega,\mathfrak F,P)$ to the family  $\mathcal F(\mathbb T)$ of closed subsets of a Polish space $\mathbb T$ is called a \emph{random closed set} if
\begin{align}
\eX^-(K)\equiv\{\omega\in \Omega:\eX(\omega)\cap K\ne\emptyset\}	
\end{align}
is in $\mathfrak F$ for each compact set $K\subseteq\mathbb T$. 
\end{definition}

\begin{definition}[Measurable Selections]
For any random set $\eX$, a measurable selection of $\eX$ is a random element $X$ with values in $\mathbb T$ such that $X(\omega)\in\eX(\omega)$ almost surely. We denote by $\Sel(\eX)$ the set of all selections from $\eX$. 	
\end{definition}
We assume one can construct a set-valued control function as a random closed set.
\begin{assumption}\label{as:cf_set}
	(i) There is a random closed set $\eV:\Omega\to \mathcal F(\cV)$ such that $V\in \eV$
	with probability 1; (ii) $\eV$ is a measurable function of observable variables and possibly a parameter $\sfs$.
\end{assumption}

The set-valued control function $\eV$ is a random closed-set constructed from the observables.\footnote{A singleton-valued control function in the literature is a special case of Assumption \ref{as:cf_set}. }  A leading case would be $\eV$ generated by a selection equation $D=\pi(Z,X,V)$ where $Z$ is a vector of instrumental variables excluded from $\sfo$ (e.g., Examples \ref{ex:Roy1}-\ref{ex:vector_D}).\footnote{In each example of Section \ref{ssec:examples}, observe that we use closed intervals or sets to construct random set $\eV$ that is closed.} However, $\eV$ can also be generated from other sources with corresponding observables and $\pi$ (e.g., Example \ref{ex:no_sel}). Assumption \ref{as:cf_set} is agnostic about the genesis of a set-valued control function. The set can depend on an unknown parameter $\sfs$, which can be infinite-dimensional. In some applications, $\sfs$ can be point identified from additional restrictions. We incorporate such restrictions into our identification framework below.
 We write $\eV(D,X,Z;\sfs)$ whenever it is useful to show its dependence on $(D,X,Z)$ and $\pi$. 

Let us discuss Assumptions \ref{as:cf_dist}-\ref{as:cf_set} further.
In the conventional CF approach, we use the control variable $V$ for two main purposes. First, we use $V$ to adjust for the effects of confounding factors as outlined in Assumption \ref{as:cf_dist}. Second, we condition on the subpopulation for which this assumption holds. We may use these properties simultaneously if $V$ is observable or can be recovered from other observable variables. However, in the current scenario, the second property is not available. Therefore, we use $\eV$ (recovered from other observables ensured by Assumption \ref{as:cf_set}) to condition on a ``coarser'' subpopulation. Failing to condition on $V$ can result in a loss of identifying power. Nevertheless, our framework enables the researcher to use all information available under the stated assumptions and establish sharp bounds on the parameters of interest.

 We represent $U$ as $U=Q(\eta;D,X,V)$ for a measurable function $Q:[0,1]^{d_U}\times\cD\times\cX\times\cV\to \cU\subseteq\mathbb R^{d_U}$ determined by the conditional distribution of $U$ given $(D,X,V)$ and a random vector $\eta\in\mathbb R^{d_U}$, which is independent of $(D,X,V)$ and is uniformly distributed over $[0,1]^{d_U}$. This representation holds generally. To see this, consider an example in which $U\equiv(U_0,U_1)$ is two dimensional as in Example \ref{ex:Roy1}. Let $(\eta_0,\eta_1)\sim U[0,1]^2.$  One can represent $(U_0,U_1)\sim F_{U|D,X,V}$ sequentially by letting
\begin{align}
	U_0&=Q_0(\eta;D,X,V)\equiv F_{U_0|D,X,V}^{-1}(\eta_0|D,X,V),\label{eq:rosenblatt1}\\
	U_1&=Q_1(\eta;D,X,V)\equiv F_{U_1|U_0,D,X,V}^{-1}(\eta_1|U_0,D,X,V),\label{eq:rosenblatt2}
\end{align} 
where for any cumulative distribution function $F$,  $F^{-1}(c)\equiv \inf \{u: F(u)> c\}$, which is the quantile function when $F$ is a continuous distribution.\footnote{This sequential transformation is known as the Knothe-Rosenblatt transform \citep[see, e.g.,][]{Villani:2008aa,Carlier:2010aa,Joe:2014wy}.}

By Assumption \ref{as:cf_dist}, we may drop $D$ from the right-hand side of \eqref{eq:rosenblatt1}-\eqref{eq:rosenblatt2}
\begin{align*}
	U_0&=Q_0(\eta;X,V)\equiv F_{U_0|X,V}^{-1}(\eta_0|X,V),\\
	U_1&=Q_1(\eta;X,V)\equiv F_{U_1|U_0,X,V}^{-1}(\eta_1|U_0,X,V).
\end{align*} 
This argument can be generalized to settings with any finite $d_U$.
In general, under Assumption \ref{as:cf_dist}, we may drop $D$ from $Q$'s argument and represent $U$ by
\begin{align}
	U=Q(\eta;X,V).~\label{eq:rosenblatt_rep}
\end{align}
The map $Q:[0,1]^{d_U}\times\cX\times\cV\to \cU\subseteq\mathbb R^{d_U}$ is determined by the conditional distribution of $U|X,V$, which we denote as
\begin{align*}
    F\equiv F_{U|X,V}.
\end{align*}
In \eqref{eq:rosenblatt_rep}, we may view $\eta$ as the remaining source of randomness in the potential outcome after controlling for $(X,V)$.

The observed outcome is determined by
\begin{align}
	Y=\sfo(D,X,U)=\sfo(D,X,Q(\eta;X,V)).\label{eq:correction}
\end{align}
One can view the right-hand side of \eqref{eq:correction} as an outcome equation augmented by an \emph{adjustment term}, $Q(\eta;X,V)$, which involves the control variable $V$ and a ``clean'' error term $\eta$ that is independent of $D$.\footnote{This is analogous to an additive model, in which the error term can be decomposed into a control function and an error term that is independent of the treatment.} 

Using the fact that $V$ is a measurable selection of $\eV$, we define the following random closed set:\footnote{Lemma \ref{lem:rdset} in the appendix establishes $\eY$ is a well-defined random closed set.}
\begin{align}
		\eY(\eta,D,X,\eV;\sfo,F)\equiv\cl\big\{y\in\cY:y=\sfo(D,X,Q(\eta;X,V)),V\in\Sel(\eV)\big\}.\label{eq:defY}
\end{align}
This set collects all outcome values (and their closure) compatible with the model structure for some unknown control variable $V$ taking values in the set-valued control function $\eV$. 
This formulation allows us to capture (i) the role of $V$ as a control variable entering the augmented outcome equation through $Q$ and (ii) model incompleteness due to the coarse information provided by $\eV$.
To our knowledge, summarizing the model prediction by a random set in \eqref{eq:defY} is new. 
 
Representing the model's prediction in this way has several advantages. First, $\eY$ collects all outcome values given all \emph{observable exogenous} variables $(D,X,\eV)$ and latent variables $\eta$. It represents
the prediction of an \emph{incomplete model} in the sense of \cite{jovanovic1989observable}.\footnote{\cite{jovanovic1989observable} characterizes an incomplete model by observed endogenous variables $y$, latent variables $\eta$, and a \emph{structure} $(\nu,\phi)$, where $\nu$ is the distribution of $\eta$, and $\phi$ is a relation such that $(y,\eta)\in\phi$. The observable exogenous variables are allowed to shift $\phi$. In our setting, $\phi$ corresponds to $gr(\eY)=\{(y,\eta):y\in \eY(\eta,d,x,\mathbf v ;\mu,F)\}$. The representation of $U$ in \eqref{eq:rosenblatt_rep} allows us to incorporate the structural parameter $(\mu,F)$ into the model's incomplete prediction ($\phi$ in \cite{jovanovic1989observable}), whereas the remaining randomness is captured by $\eta\sim U[0,1]^{d_U}$.}   Following the partial identification literature, we systematically obtain sharp identifying restrictions in such models in the next section. Second, $\eY$ builds on an augmented outcome equation, which often helps derive closed-form bounds; e.g., see the discussion of the next paragraph. Finally, the framework can accommodate both continuous and discrete outcomes. We provide further details in Section \ref{sec:applications}.

\begin{remark}
Assumption \ref{as:cf_dist} plays an important role in obtaining identifying restrictions for structural parameters via $\eY$. Each measurable selection of $\eY$ is represented by $Y=\mu(D,X,Q(\eta;X,V)),$ which separates the model into two components with distinct roles.  The structural function $\mu$ captures the effect of the treatment $D$, while the adjustment term $Q$ depends only on $(X,V)$ and unobserved heterogeneity $\eta$, but not on $D$.  This separation, made possible by Assumption \ref{as:cf_dist} and the representation of $U$ via $Q$, clarifies that variation in $D$ affects $Y$ solely through $\mu$ once $(X,V)$ are held fixed. In this sense, $Q$ acts as a control that isolates the structural effect of $D$. 

Since $Q$ does not depend on $D$, it also facilitates recovering structural parameters. For example, we may 
express structural quantities such as the \emph{average conditional response} $E[Y(d)|X=x,V=v]$ by integrating out $\eta$
\begin{align}
E[Y(d)|X=x,V=v]=  \int_{[0,1]^{d_U}} \mu(d,x,Q(\eta;x,v))d\eta.
\end{align}  
After characterizing the sharp identification region for $\theta$,
we use this property to obtain bounds on various structural functions of interest (see Section \ref{ssec:functionals}).
\end{remark}

\section{Identification}
Let $P_0$ be the joint distribution of the observable variables $(Y,D,X,Z)$. Let $\theta\equiv(\sfo,F,\sfs)$ collect the structural parameters.
 Let $\Theta\equiv\mathsf M\times\mathsf F\times \mathsf \Pi$ be the parameter space for $\theta$, which embodies a priori restrictions on the parameter. As discussed earlier, some models provide additional restrictions on $\sfs$.\footnote{Consider Example \ref{ex:Roy1}. Under an additional independence assumption $Z\perp V|X$, $\Pi_r(P_0)=\{\sfs\in\Pi:\sfs(z,x)=P_0(D=1|Z=z,X=x)\}$.} We let $\mathsf \Pi_r(P_0)\subset \mathsf \Pi$ be the set of selection parameters satisfying them. 

We define the sharp identification region for $\theta$ as follows.
\begin{definition}[Sharp Identification Region under Full Independence]
The \emph{sharp identification region} $\Theta_I(P_0)\subset  \mathsf M\times\mathsf F\times \mathsf \Pi_r(P_0)$  is a subset of $\Theta$ such that each of its elements $\theta=(\sfo,F,\sfs)$ satisfies the following statements: (i) For any $Y\sim P_0(\cdot|D,X,Z)$, one can represent the outcome as $Y=\sfo(D,X,U)$ for some $U$ whose conditional law $F$ satisfies Assumption \ref{as:cf_dist} for some $V:\Omega\to\cV$. (ii) The control variable $V$ is a measurable selection of a set-valued control function $\eV$ satisfying Assumption \ref{as:cf_set}.
\end{definition}
The main result (Theorem \ref{thm:identification}) of this section characterizes $\Theta_I(P_0)$ through inequality restrictions on $\theta$. For this, we introduce the \emph{containment functional} $\contf$ of the random set $\eY$. For any closed set $A\subset \cY$ and $(d,x,z)\in\cD\times\cX\times\cZ$,  let
\begin{equation}
	\contf(A|D=d,X=x,Z=z)
 \equiv \int_{[0,1]^{d_U}}1\big\{\eY(\eta,d,x,\eV(d,x,z);\sfo,F)\subseteq A\big\}d\eta.\label{eq:containment}
\end{equation}
 This functional uniquely determines the distribution of $\eY$ \citep{Molchanov:2017th}. Since $\eta$ is uniformly distributed over $[0,1]^{d_U}$, the right-hand side of \eqref{eq:containment} can be computed analytically or by simulation (see Section \ref{sec:applications}). 
 
The containment functional $\contf$ characterizes the distribution of \emph{all} measurable selections of $\eY$ in the following sense:\footnote{This equivalence holds up to an ordered coupling \citep[][Chapter 2]{Molchanov:2018ui}.}
 \begin{multline}
Y \in \eY(\eta,d,x,\eV(d,x,z);\sfo,F),~ a.s.\\
~\Leftrightarrow~P_0(Y\in A|d,x,z)\ge \contf(A|d,x,z), ~\forall A\in\cF(\cY),~(d,x,z)-a.s.\label{eq:artstein}
\end{multline} 
This inequality restriction is known as \emph{Artstein's inequality} \citep[see, e.g.,][Theorem 2.13]{Molchanov:2018ui}, which is the central device to derive sharp identifying restrictions in incomplete models. 

Under Assumption \ref{as:cf_set}(i), the observed outcome $Y$ is a measurable selection of $\eY$, and Artstein's inequality relates the distribution of $Y$ with the distribution of $\eY$. The left-hand side $P_0(Y\in A|D,X,Z)$ of the inequality can be recovered from a large sample of the observable variables $(Y,D,X,Z)$. The right-hand side $\contf(A|D,X,Z)$ can be computed from model primitives. Taken together, they provide identifying restrictions. We illustrate them through examples (see Sections \ref{ssec:ex_dynamic4} and \ref{ssec:ex_multinomial}). 

By the definition of $\eY$, Artstein's inequality is equivalent to the existence of $Y$ such that
\begin{align}
Y=\mu(D,Q(\eta,V))=\mu(D,U)~,
\end{align}
for some $(U,V)$ satisfying Assumptions \ref{as:cf_dist}-\ref{as:cf_set}, ensuring the sharpness of the restriction.

The following theorem characterizes the sharp identification region.
\begin{theorem}\label{thm:identification}
Suppose Assumptions \ref{as:cf_dist}-\ref{as:cf_set} hold. Then, the sharp identification region for the structural parameter $\theta=(\sfo,F,\sfs)$ is
\begin{align}
	\Theta_I(P_0)=\{\theta\in\Theta:P_0(Y\in A|D,X,Z)\ge \contf(A|D,X,Z),~a.s.
	~\forall A\in \cF(\cY),~\sfs\in \mathsf \Pi_r(P_0)\}.
\end{align}	
\end{theorem}

Artstein's inequality is introduced to the identification literature in econometrics by \cite{galichon2011set} and has been extensively used. As in other work, we use this result to convert the model's set-valued prediction into a system of inequality restrictions that do not involve the unobserved control variable $V$, making the resulting restrictions amenable to estimation.  Practitioners can use \eqref{eq:artstein} to make inference for the elements of $\Theta_I(P_0)$ or their functions.
For example, one may use inference methods for conditional moment inequalities \citep{Andrews:2013aa,Chernozhukov:2013aa} or likelihood-based inference methods \citep{Chen_2018,Kaido:2022aa}. We provide an empirical illustration utilizing a likelihood-based inference method in Section \ref{sec:empirical_illustration}.

\begin{remark}
For a given $(d,x,z)$, the number of the inequalities in \eqref{eq:artstein} is finite as long as $\cY$ is a finite set. Furthermore, it often suffices to impose a subset of inequalities to characterize $\Theta_I(P_0)$. Such a subset $\cA\subseteq\cF(\cY)$ is called the \emph{core determining class} \citep{galichon2011set}. The smallest core determining class only depends on support of $\eY(\cdot,D,X,Z;\sfo,F)$ and does not depend on $P_0$ \citep{Luo:2017ab,Ponomarev22thesis}.\footnote{\cite{Ponomarev22thesis} provides an algorithm based on the connectedness of suitable subgraphs to determine the smallest core determining class. See also \cite{Chesher:2017vu,BONTEMPS2020373}.}

If $Y$ is continuous, \eqref{eq:artstein} involves infinitely many inequalities.\footnote{While we do not pursue this here, an alternative approach to deal with continuous variables would be to characterize the sharp identification region through an optimal transport problem instead of Artstein's inequalities \citep[see][]{li2024finitesampleinferenceincomplete}.} Nevertheless, under a weaker conditional mean independence assumption, a commonly used empirical specification admits a sharp characterization with finitely many inequalities; see Section \ref{sec:mean_indep}.
\end{remark}

\subsection{Conditional Mean Restrictions} \label{sec:mean_indep}
So far, we worked with the control function assumption in the form of conditional independence assumption (Assumption \ref{as:cf_dist}). A weaker conditional mean independence assumption is also considered in the literature  \citep[see e.g.][]{Newey:1999tu,Pinkse:2000vg}. This section explores identifying restrictions that can be obtained from this assumption.

We focus on a scalar outcome $Y$.   Let $U\equiv(U_d,d\in\cD)$ and $U_D=\sum_{d\in \cD}U_d1\{D=d\}$.\footnote{For continuous $D$, we may define $U_D$ by $U_D=\int_{\cD} U_{s} d\delta_{D}(s)$, where $\delta_D$ is a Dirac measure at $D$.} Consider the following additive model:
\begin{align}
	Y=\sfo(D,X)+U_D,\label{eq:separable_mu}
\end{align}
This way, $Y$ is a function of vector $U$, making this model a special case of \eqref{eq:outcome_eq}. It nests the linear model $\mu(d,x)=\alpha d+x'\beta$ with scalar unobservable $U$ as a special case.\footnote{A similar argument can be applied to a nonadditive model $Y=\sfo(D,X,U)$ for which $U$ is a scalar and $\sfo$ is invertible with respect to $U$. We focus on the additive model only for notational simplicity.}  

Suppose the following  mean independence analog of Assumption \ref{as:cf_dist} holds.
\begin{assumption}\label{as:cf_mean}
For each $d\in\cD$, $E[|U_d|]<\infty$, and $E[U_d|D,X,V]= E[U_d|X,V], ~a.s.$
\end{assumption}

For each $d\in\cD$, let $\lambda_d(X,V)\equiv E[U_d|X,V]$ and $\eta_d\equiv U_d-E[U_d|X,V].$ Under Assumption \ref{as:cf_mean}, we may write
\begin{align}
	E[Y|D=d,X=x,V=v]=\sfo(d,x)+\lambda_d(x,v)~,
\end{align}
and $Y=\sfo(D,X)+\lambda_d(X,V)+\eta_d$. We note that $\lambda_d$ is a known function of $F$ and plays the role of the adjustment term similarly to $Q$.

We also assume that $U$ is continuously distributed, which helps us derive tractable identifying restrictions. This assumption can be dropped when $\eY$ defined below is interval-valued almost surely.
\begin{assumption}\label{as:Fuv}
$U|D,X,V$ has a strictly positive density with respect to Lebesgue measure on $\mathbb R^{d_U}$ almost surely. 
\end{assumption}

We now define the sharp identification region as follows.
\begin{definition}[Sharp Identification Region under Mean Independence]
The \emph{sharp identification region under mean independence} $\Theta_I(P_0)\subset  \mathsf M\times\mathsf F\times \mathsf \Pi_r(P_0)$  is a set such that each $\theta=(\sfo,F,\sfs)\in \Theta_I(P_0)$ satisfies the following statement: (i) For any $Y$ whose conditional mean is $E_{P_0}[Y|D,X,Z]$, one can represent the outcome as in \eqref{eq:separable_mu}, where $U$'s conditional law $F$ satisfies Assumptions \ref{as:cf_mean} and \ref{as:Fuv} for some $V:\Omega\to\cV$. (ii) The control variable $V$ is a measurable selection of a set-valued control function $\eV$ satisfying Assumption \ref{as:cf_set}.	
\end{definition}

Let $\eta\equiv(\eta_d,d\in\cD)$. Define
\begin{align}
	\eY(\eta,D,X,Z;\mu,F)\equiv\cl\big\{y\in\cY:y=\sfo(D,X)+\lambda_D(X,V)+\eta_D,V\in\Sel(\eV)\big\}.\label{eq:defeY_mean}
\end{align}
Since $Y\in \Sel \eY$ by Assumption \ref{as:cf_set}, the observed conditional mean $E_{P_0}[Y|D,X,Z]$ belongs to the set of the conditional mean of measurable selections of $\eY(\eta,D,X,Z;\mu,F)$ for some $\theta\in\Theta_I(P_0)$. To use this observation, we introduce the \emph{conditional Aumann expectation} of a random set. 

A random closed set $\eX$ is said to be \emph{integrable} if $\eX$ has at least one integrable selection.
We define the Aumann (or selection) expectation of an integrable random closed set following \cite{beresteanu2011sharp} and \cite{Molinari:2020un}. For this, we let $\Sel^1(\eX)$ denote the set of integrable selections of $\eX$.
\begin{definition}\label{def:aumann}
The \emph{Aumann expectation} of an integrable random closed set $\eX$ is given by
\begin{align}
\mathbb E[\eX]\equiv\cl\Big\{E[X]:,X\in \Sel^1(\eX)\Big\}.	
\end{align}
For each sub $\sigma$-algebra $\mathfrak B\subset \mathfrak F$, the conditional Aumann expectation of $X$ given $\mathfrak B$ is the $\mathfrak B$-measurable random closed set $\boldsymbol{R} \equiv \mathbb E(\eX|\mathfrak B)$ such that the family of $\mathfrak B$-measurable integrable selections of $\boldsymbol{R}$, denoted $\Sel^1_{\mathfrak B}(\boldsymbol{R})$, satisfies
\begin{align}
\Sel^1_{\mathfrak B}(\boldsymbol{R})\equiv\cl\Big\{E[X|\mathfrak B]:,X\in \Sel^1(\eX)\Big\},
\end{align}
where the closure in the right-hand side is taken in $L^1$.
\end{definition}

Under the maintained assumptions, the model's prediction is summarized by 
\begin{align}
E_{P_0}[Y|D,X,Z]\in \mathbb E[\eY(\eta,D,X,Z;\sfo,F)|D,X,Z],~a.s.	
\end{align}
This condition is equivalent to
\begin{align}
bE_{P_0}[Y|D,X,Z]\le s(b, \mathbb E[\eY(\eta,D,X,Z;\sfo,F)|D,X,Z])	,~b\in\{-1,1\},\label{eq:suppfunc_restriction}
\end{align}
where $s(b,K)=\sup_{k\in K}bk$ is the \emph{support function} of $K$. 
As pointed out in the literature, directly working with the conditional Aumann expectation operator can be computationally demanding.
We therefore use Assumption \ref{as:Fuv} to ensure the convexification property \citep[][Theorem A.2.]{Molinari:2020un} of $\mathbb E[\eY(\eta,D,X,Z;\sfo,F)|D,X,Z]$.
The convexification property allows us to interchange the expectation and support function operations and obtain the tractable restriction in the following theorem.

\begin{theorem}\label{thm:mean_id}
Suppose Assumptions \ref{as:cf_set}-\ref{as:Fuv} hold. Suppose $E_{P_0}[|Y|]<\infty$. 
Then,  the sharp identification region is
\begin{multline}
	\Theta_I(P_0)=\big\{\theta\in\Theta:\sfo(d,x)+\lambda_L(d,x,z)\le E_{P_0}[Y|D=d,X=x,Z=z]\\
	\le \sfo(d,x)+\lambda_U(d,x,z),~\sfs\in \mathsf \Pi_r(P_0)\big\},\label{eq:aumann}
\end{multline}	
where 
\begin{align}
	\lambda_L(d,x,z)\equiv\inf_{v\in \eV(d,x,z;\sfs)}\lambda_d(x,v),~~\lambda_U(d,x,z)\equiv\sup_{v\in \eV(d,x,z;\sfs)}\lambda_d(x,v).
\end{align}
\end{theorem}
 
\subsection{Causal and Counterfactual Objects}\label{ssec:functionals}
Based on Theorems \ref{thm:identification} or \ref{thm:mean_id}, one can construct bounds on functionals of $\theta$. Let $W\equiv (X,V)$ and let $F_W$ be its distribution. Given  $\varphi:\mathbb R\to\mathbb R$, let \begin{align}
	\kappa(d) \equiv E[\varphi(Y(d))]=\int\int\varphi(\sfo(d,x,Q(\eta;w)))d\eta dF_W(w).\label{eq:kappa}
\end{align}
The average and distributional structural functions are special cases of $\kappa$.

For the average structural function $\text{ASF}(d)\equiv E[\sfo(d,X,U)]=E[Y(d)]$ considered by \cite{Blundell:2003wi}, we may set $ \varphi(Y(d))=Y(d)$, which yields
\begin{align}
	\text{ASF}(d)\equiv\int \int\sfo(d,x,u)dF(u|w) dF_W(w)
	=\int \int\sfo(d,x,Q(\eta;w))d\eta dF_W(w).
\end{align}
The average treatment effect (ATE) is then $\text{ATE}(d,d')=\text{ASF}(d)-\text{ASF}(d')$.
For the distributional structural function \citep{Chernozhukov:2020aa}, we may set $\varphi(Y(d))=1\{Y(d)\le y\}$, which gives
\begin{multline}
	\text{DSF}(y,d)\equiv \int\int 1\{\sfo(d,x,u)\le y\}dF_U(u|w)dF_W(w)\\
=\int \int 1\{\mu(d,x,Q(\eta;w))\le y\}d\eta dF_W(w).
\end{multline}
The \emph{quantile structural function} (QSF), the $\tau$-th quantile of $Y(d)$, can be obtained using $\text{QSF}(d)\equiv\text{DSF}^{-1}(\tau,d)$ \citep{Imbens:2002vs}.

 The following proposition characterizes the identification region for $\kappa$.
\begin{theorem}\label{thm:functional}
Suppose the conditions of Theorem \ref{thm:identification} or \ref{thm:mean_id} hold. Suppose $\varphi$ is bounded, and the underlying probability space is non-atomic. Then, the sharp identification region for $\kappa$ is
\begin{align}
	\mathfrak K_I(d)=\bigcup_{\theta\in\Theta_I(P_0)}[\underline{\kappa}(d;\theta), \overline{\kappa}(d;\theta)],
\end{align}
where
\begin{align}
\overline{\kappa}(d;\theta)&\equiv	E[\sup_{v\in \eV(D,X,Z;\sfs)}\int\varphi(\sfo(d,X,Q(\eta;X,v)))d\eta],\\
\underline{\kappa}(d;\theta)&\equiv E[\inf_{v\in \eV(D,X,Z;\sfs)}\int\varphi(\sfo(d,X,Q(\eta;X,v)))d\eta],
\end{align}
and the expectation above is taken with respect to the distribution of $(D,X,Z)$.
\end{theorem}

The identification region for $\kappa$ is expressed as a union of intervals. Practically, one may only be interested in the upper and lower endpoints of $\mathfrak K_I(d)$. They are given by the following corollary.
\begin{corollary}
Suppose the conditions of Theorem \ref{thm:functional} hold. Then, the tight upper and lower bounds of $\mathfrak K_I(d)$ are   \begin{align}
\overline{\kappa}(d)&\equiv\sup_{\theta\in \Theta_I(P)}E[\sup_{v\in \eV(D,X,Z;\sfs)}\int\varphi(\sfo(d,X,Q(\eta;X,v)))d\eta],\label{eq:functional_bounds1}\\
\underline{\kappa}(d)&\equiv	\inf_{\theta\in \Theta_I(P)}E[\inf_{v\in \eV(D,X,Z;\sfs)}\int\varphi(\sfo(d,X,Q(\eta;X,v)))d\eta].\label{eq:functional_bounds2}
\end{align}
If $F_W$ is point identified, the tight upper and lower bounds are 
\begin{align}
\overline{\kappa}(d)&\equiv\sup_{\theta\in\Theta_I(P)}E_{\eta,W}[\varphi(\sfo(d,X,Q(\eta;W)))]\label{eq:functional_bounds3}\\
\underline{\kappa}(d)&\equiv\inf_{\theta\in\Theta_I(P)}E_{\eta,W}[\varphi(\sfo(d,X,Q(\eta;W)))].\label{eq:functional_bounds4}
\end{align}
\end{corollary}
Theorem \ref{thm:functional} does not presume point identification of $F_W$ because $V$ is unobservable, which leads to general bounds in \eqref{eq:functional_bounds1}-\eqref{eq:functional_bounds2}.
In some examples, $F_W$ is point identified even if $V$ itself is not uniquely recovered.\footnote{In Example \ref{ex:Roy1}, the distribution of $V$ is normalized to $U[0,1]$.} If so, for each $\theta\in\Theta_I(P_0)$,
$\overline{\kappa}(d;\theta)=\underline{\kappa}(d;\theta).$ This allows us to simplify the bounds as in \eqref{eq:functional_bounds3}-\eqref{eq:functional_bounds4}.

In addition to the structural parameter \eqref{eq:kappa}, one can consider a policy that only changes the selection behavior. Suppose a policy sets $Z$ (e.g., a tuition subsidy) to $z$, and the treatment selection under this policy is $D(z)=\sfs(z,X,V)$. The \emph{policy-relevant structural function (PRSF)} would be
\begin{align}
    \kappa(z)\equiv E[\varphi(Y(D(z))]=\int\int\varphi(\sfo(\sfs(z,w),x,Q(\eta;w)))d\eta dF_W(w).
\end{align}
The PRSF is related to the policy-relevant treatment effect (PRTE) and marginal PRTE introduced in \cite{Heckman:2005aa} and \cite{carneiro2010evaluating}.

One can consider another related structural function. Suppose $D=(D_1,D_2)$, and let $Y(d_1,d_2)$ denote the counterfactual outcome given $(d_1,d_2)$ and $D_{2}(d_1)$ denote the counterfactual treatment of $D_2$ given $d_1$. Then the \emph{mediated structural function (MSF)} would be
\begin{align}
    \kappa(d_1,d_1')\equiv E[\varphi(Y(d_1,D_{2}(d_1')))]=\int\int\varphi(\sfo(d_1,\sfs_2(d_1',z,w),x,Q(\eta;w)))d\eta dF_{Z,W}(z,w),
\end{align}
where we allow $d_1\neq d_1'$. The MSF can be used to define the direct causal effect of one treatment and the indirect causal effect mediated by another treatment. This scenario is relevant in Example \ref{ex:vector_D} on strategic interaction (e.g., a player's decision being mediated by the opponent's decision) and on dynamic treatment effects (e.g., a previous treatment being mediated by the previous outcome; \citealp{han2023semiparametric}).
One can derive bounds on these objects in a similar manner.

\section{Applications of the Identification Results}\label{sec:applications}
We illustrate the use of Theorems \ref{thm:identification} and \ref{thm:mean_id} through examples.\footnote{In the illustrations, we pair continuous outcome variables with the generalized Roy model and strategic treatment decisions. We pair discrete outcomes with other examples. These choices are arbitrary.
Theorems \ref{thm:identification} and \ref{thm:mean_id} allow the researcher to combine various outcome variable types, selection models, and other sources of controls.} We present examples with various selection processes (Sections \ref{ssec:roy}-\ref{ssec:ex_strategic}) and an example with an incomplete control (Section \ref{ssec:ex_school}). Appendix \ref{sec:additional_examples} provides further examples.

\subsection{Generalized Roy Model with a Continuous Outcome}\label{ssec:roy}
We revisit Example \ref{ex:Roy1}. Let $U\equiv (U_{1},U_{0})$, and recall that 
\begin{align*}
	D=1\{\sfs(Z,X)\ge V\};
\end{align*}
hence $U$'s conditional mean independence from $D$ holds as long as $U$ is mean independent of the instrument $Z$. Suppose $Y$ satisfies \eqref{eq:separable_mu}.
Let  $\lambda_d(X,V)\equiv E[U_d|X,V]$ for $d\in\cD$, and let 
\begin{align}
	\lambda_L(d,x,z)\equiv\inf_{v\in \eV(d,x,z;\sfs)}\lambda_d(x,v),~~\lambda_U(d,x,z)\equiv\sup_{v\in \eV(d,x,z;\sfs)}\lambda_d(x,v).
\end{align}
The model's prediction is
\begin{align}
	\eY(\eta,D,X,\eV;\sfo,F)\equiv\cl\{y\in\cY:y=\sfo(D,X)+\lambda_D(X,V)+\eta_D,V\in\Sel(\eV)\}.\label{eq:Roy2eY}
\end{align}
By Theorem \ref{thm:mean_id}, we obtain the following inequalities:
\begin{align}
	 E_{P_0}[Y|D=d,X=x,Z=z]&\le \sfo(d,x)+\lambda_U(d,x,z)\label{eq:ex5_1}\\
	 E_{P_0}[Y|D=d,X=x,Z=z]&\ge \sfo(d,x)+\lambda_L(d,x,z).\label{eq:ex5_2}
\end{align}	
Rearranging them and taking their intersections across $z$ give the following result.
\begin{corollary}\label{cor:ex1}
Suppose $E_{P_0}[|Y|]<\infty$. Suppose  $U_0,U_1|X,Z$ have a density with respect to Lebesgue measure, and $E[U_d|Z,X,V]=E[U_d|X,V],d=0,1$. 
Then, $\Theta_I(P_0)$ is the set of parameter values $\theta \equiv (\mu,F,\pi)$ such that, for almost all $(d,x,z)$,
\begin{multline}
 \sup_{z\in\cZ}\Big\{E_{P_0}[Y|D=d,X=x,Z=z]-\lambda_U(d,x,z)\Big\}	\\
 \le \sfo(d,x)\le \\
 \inf_{z\in\cZ}\Big\{E_{P_0}[Y|D=d,X=x,Z=z]-\lambda_L(d,x,z)\Big\}\label{eq:Roy1_intersection_bounds}.
\end{multline}
\end{corollary}
The identifying restrictions \eqref{eq:Roy1_intersection_bounds} take the form of intersection bounds on $\mu$. For each $z$, $E_{P_0}[Y|D=d,X=x,Z=z]-\lambda_U(d,x,z)$ defines a lower bound on $\mu(d,x)$. Since $z$ is excluded
from $\mu$, we can intersect the lower bounds across all values of $z$.  The upper bound is formed similarly. 

It is worth noting that  \eqref{eq:Roy1_intersection_bounds} restricts the parameter vector $\theta \equiv(\sfo,F,\sfs)$ jointly because $\lambda_L,\lambda_U$ are functions of $(F,\sfs)$. Therefore, they are also useful for bounding $(F,\sfs)$. Furthermore, if $Z\perp V|X$, $\sfs$ is point identified as the propensity score $\pi(z,x)=P_0(D=1|Z=z,X=x)$. Hence, in this case, \eqref{eq:Roy1_intersection_bounds} gives joint restrictions on $(\sfo,F)$. 

\begin{remark}
The terms $\lambda_U,\lambda_L$ can be seen as \emph{adjustment terms} to account for the effects of $V$. To see this, suppose $\eV$ is a singleton $\{V(D,X,Z;\sfs)\}$ (e.g., because $D=\pi(Z,X)+V$). Then,
\begin{align}
\lambda_L(d,x,z)=\lambda_U(d,x,z)=\lambda_d(x,z)=
E[U_d|X=x,V=v],
\end{align}
In this case,  \eqref{eq:ex5_1}-\eqref{eq:ex5_2} reduce to 
\begin{align}
E[Y|D=d,X=x,Z=z]=\sfo(d,x)+E[U_d|X=x,V=v].
\end{align}
Hence, it justifies regressing $Y$ on $(D,X)$ with an additive correction term \citep{Newey:1999tu}. This argument works only when $\eV$ is singleton-valued. In the general setting with a set-valued control function, one can work with the intersection bounds in \eqref{eq:Roy1_intersection_bounds}.
\end{remark}

\subsection{Dynamic Treatment Effects}\label{ssec:ex_dynamic4}
We consider a model of dynamic treatment decisions with imperfect compliance 
\citep{Robins:1997aa,Han:2021aa}. 
In the initial period, binary treatment $D_1$ (e.g., a medical treatment) and binary outcome $Y_1$ (e.g., the presence of symptoms)  are generated according to
	\begin{align}
D_{1} & =1\{\sfs_{1}(Z_{1},X)\ge V_{1}\},\label{eq:ex_dynamic1} \\
Y_{1} & =1\{\mu_{1}(D_{1},X)\ge U_{1}\}.\label{eq:ex_dynamic2}
\end{align}
In the next period,
the observed treatment status is determined based on the initial treatment and outcome:
\begin{align}
D_{2} & =1\{\sfs_{2}(Y_{1},D_{1},Z_{2},X)\ge V_{2}\}.\label{eq:ex_dynamic3}
\end{align}
Finally, the outcome in period 2 is determined by
\begin{align}
Y_{2} & =1\{\mu_{2}(Y_{1},D_{1},D_{2},X)\ge U_{2}\}.\label{eq:ex_dynamic4}
\end{align}
For each $t$, $U_{t}$ and $V_{t}$ are normalized to $U[0,1]$ conditional
on $X=x$. 

Consider the effect of the initial outcome and treatment history $D=(Y_1,D_1,D_2)$ on $Y_2$. One may be concerned about endogeneity because $U_2$ may depend on  $(U_1,V_1,V_2)$. For example, $U_1$ and $U_2$ may share a time invariant component. Another possibility is that $U_2$ may be related to $(V_1,V_2)$ through the agent's dynamic treatment take-up decisions.

Below, we let $Y\equiv Y_2$, $U\equiv U_2$ and let $V\equiv(U_1,V_1,V_2)$ be unobserved control variables; also let $Z\equiv(Z_1,Z_2)$, and let $\sfs\equiv(\mu_1(\cdot),\sfs_1(\cdot),\sfs_2(\cdot))$. 
 Inspecting the system of selection equations, 
the assignment of $D=(Y_1,D_1,D_2)$ is independent of $U_2$ conditional on $(X,V)$ as long as the instrumental variables $Z$ are independent of $U_2$.

For notational simplicity, we rewrite \eqref{eq:ex_dynamic4} as
\begin{align}
	Y=1\{\sfo(D,X)\ge U\},
\end{align}
and derive the model prediction $\eY$ as follows.
First, let $U=Q(\eta|X,V)=F^{-1}(\eta|X,V)$.	
Then,
\begin{align}
	Y&=1\{\sfo(D,X)\ge Q(\eta|X,V)\}\notag\\
	&=1\{F(\sfo(D,X)|X,V)\ge\eta\}\notag\\
	&=1\{H(D,X,V)\ge \eta\},\label{eq:dynamic_augout}
\end{align}
 where $H(d,x,v)\equiv F(\sfo(d,x)|x,v)$ is the average response conditional on $(x,v)$. 

Next, the dynamic selection and outcome equations allow us to restrict $V$ to the following set:
\begin{align}
	\eV(D,Z,X;\pi)=\eV_{U_1}(D,X;\sfo_1)\times \eV_{1}(D,Z_1,X;\sfs_1)\times \eV_{2}(D,Z_2,X;\sfs_2),~\label{eq:ex_dynamic_eV}
\end{align}
where 
\begin{align*}
	&\eV_{U_1}(D,X;\sfo_1)\equiv\begin{cases}
		[\mu_1(D_1,X),1] & \text{if }  Y_1=0\\
		[0,\mu_1(D_1,X)] & \text{if }  Y_1=1,
	\end{cases}~~
	\eV_{1}(D,Z_1,X;\sfs_1)\equiv\begin{cases}
		[\sfs_1(Z_1,X),1] & \text{if }  D_1=0\\
		[0,\sfs_1(Z_1,X)] & \text{if }  D_1=1,
	\end{cases}\\
	&\eV_{2}(D,Z_2,X;\sfs_2)\equiv\begin{cases}
		[\sfs_2(Y_1,D_1,Z_2,X),1] & \text{if }  D_2=0\\
		[0,\sfs_2(Y_2,D_2,Z_2,X)] & \text{if }  D_2=1.
	\end{cases}
\end{align*}
 By the augmented outcome equation \eqref{eq:defY} and \eqref{eq:dynamic_augout},
we obtain the following model prediction:
 \begin{align}
 \eY(\eta,D,X, \eV;\sfo,F)=\begin{cases}
 	\{0\} & \eta> \sup_{V\in\Sel(\eV)}H(D,X,V)\\
 	\{0,1\} & \inf_{V\in\Sel(\eV)}H(D,X,V)<\eta\le \sup_{V\in\Sel(\eV)}H(D,X,V)\\
 	\{1\} &\eta \le  \inf_{V\in\Sel(\eV)}H(D,X,V).
 \end{cases}\label{eq:eYrep}	
 \end{align}
This expression exhibits an  \emph{incomplete threshold-crossing structure} shown in Figure \ref{fig:inc_cf1}.\footnote{The incomplete threshold-crossing structure also appears in semiparametric binary choice models with interval-valued covariates \citep{Manski:2002um}. \citeposs{Manski:2002um} model is considerably different from ours; they consider $Y=1\{W'\theta+\delta X^{*}+\epsilon>0\}$, where $W$ is exogenous, $X^{*}$ is observed as an interval (i.e. $X^{*}\in [X_L,X_U]$), $\delta>0$ and $\epsilon$ satisfies a quantile independence condition. See also \cite{Molinari:2020un} (Section 3.1.1) for an extensive discussion of their model.}
If $\eta$ is below the lower threshold, the model predicts $\eY=\{1\}$, whereas $\eY=\{0\}$ if $\eta$ is above the upper threshold. The model predicts $\eY=\{0,1\}$ if $\eta$ is between the two thresholds.
\begin{figure}[htbp]
\begin{center}
\begin{tikzpicture}[scale=0.8, domain=0:15,>=latex]
	\draw[->] (0,0) -- (0,7) node[left] {$\eta$};
	\draw[-] (0,0) -- (6,0);
	\draw[dashed] (0,2.4) -- (6,2.4) ;
	\draw[dashed] (0,4) -- (6,4) ;	
	\draw (0,2.4) node[anchor=east] {$  \inf_{v\in\eV(d,x,z;\sfs)}H(d,v)$};
	\draw (3,1.2) node {$ \{1\}$};
	\draw (3,3.2)   node { $\{0,1\}$};
	\draw (3,5)   node { $\{0\}$};
	\draw (0,4) node[anchor=east] {$  \sup_{v\in\eV(d,x,z;\sfs)}H(d,v)$};
	\draw (0,0) node[anchor=east] {$0$};
	\draw (0,6) -- (6,6);
	\draw (0,6) node[anchor=east] {$1$};
\end{tikzpicture}
\caption{An incomplete threshold-crossing structure.}
\label{fig:inc_cf1}
{\footnotesize Note: The figure shows the value of $\eY(\eta|d,\eV;\nu,F)$ as a function of $\eta$. }
\end{center}
\end{figure}
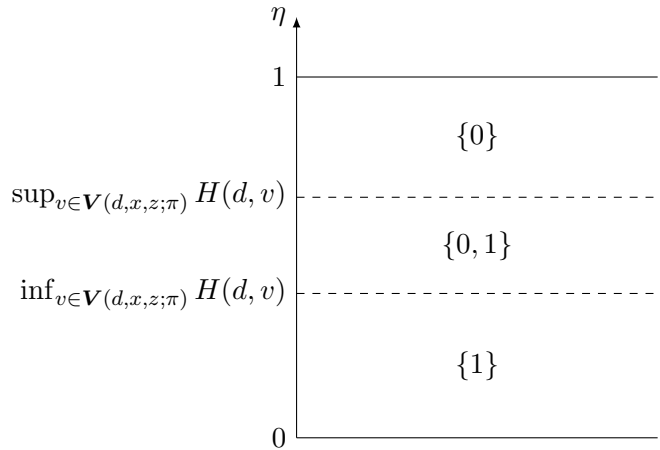

The containment functional of $\eY$ in \eqref{eq:eYrep} satisfies
\begin{align*}
 \contf(\{1\}|D=d,X=x,Z=z)&=F_\eta(\eY(\eta,D,X,\eV;\sfo,F)\subseteq \{1\}|D=d,X=x,Z=z)\\
 &=\inf_{v\in\eV(d,x,z;\sfs)}H(d,x,v),\\
 \contf(\{0\}|D=d,X=x,Z=z)&=F_\eta(\eY(\eta,D,X,\eV;\sfo,F)\subseteq \{0\}|D=d,X=x,Z=z)\\
 &=1-\sup_{v\in\eV(d,x,z;\sfs)}H(d,x,v).	
\end{align*}
Theorem \ref{thm:identification} then implies simple identifying restrictions:
\begin{align}
\inf_{v\in\eV(d,x,z;\sfs)}H(d,x,v)\le P(Y=1|D=d,X=x,Z=z)\le 	\sup_{v\in\eV(d,x,z;\sfs)}H(d,x,v).\label{eq:dynamic_ineq1}
\end{align}

We may apply this argument sequentially to \eqref{eq:ex_dynamic2}--\eqref{eq:ex_dynamic4}. The next step is to take $Y=D_2$ as an outcome, $D=(Y_1,D_1)$ as a treatment, $U=V_2$ as a latent variable in the outcome equation, and $V=(V_1,U_1)$ as control variables, which generates inequalities of the form  \eqref{eq:dynamic_ineq1}. Finally, we may apply the same argument to the outcome and selection equations in period 1. Corollary \ref{cor:ex4} in the Appendix characterizes the sharp identification region for $\theta$.

\subsection{Treatment Responses with Social Interactions}\label{ssec:ex_strategic}

 We consider settings where other individuals' treatment status affects one's outcome through spillover or equilibrium effects. Let individuals be indexed by $j=1,\dots,J$. Let $D=(D_{1},\dots,D_{J})$
be a vector of treatment decisions across individuals and let $D_{-j}$ be the vector $D$ without the element $D_{j}$.
Consider the effect of the entire profile $D$ on some outcome $Y$. For example, $D$ indicates entries of potential market participants (e.g., airlines) and  $Y$ is a market-level outcome (e.g., pollution).
Suppose the observed
treatments $D$ satisfy 
\begin{align}
D_{j} & =1\{\sfs_{j}(D_{-j},Z_{j},X)\ge V_{j}\},~j=1,\dots,J,\label{eq:strategic_selection}
\end{align}
where $V_{j}|X=x$ is normalized to $U[0,1]$.\footnote{The joint distribution of $V=(V_1,\dots,V_J)$ is unrestricted.} 

Many empirical settings require relaxing the \emph{Stable Unit Treatment Value Assumption (SUTVA)} \citep{Rubin1973} (or, equivalently, the \emph{Individualistic Treatment Response (ITR)} of \cite{Manski:2013aa}), particularly when outcomes are shaped by interactions across individuals. 
One way to motivate \eqref{eq:strategic_selection} is to allow for such interdependence through strategic selection. Let $Y_{j}(d_{1},\dots,d_{J})$
be the potential outcome of individual $j$ when $D$ is set to $(d_{1},\dots,d_{J})$.
Previous examples assume an individual's outcome only depended
on their own treatment (i.e., SUTVA or ITR) that $Y_{j}(d_{1},\dots,d_{J})=Y_{j}(d_{j})$. We istead
allow each individual's outcome to depend on the entire vector of
treatments received by the individuals, capturing spillover effects that are central in many empirical applications  \citep{Graham:2011aa,Aronow:2017aa}.

For simplicity, consider two individuals. Each 
faces a binary action $d_j$. For individual $j$ and $(d_1,d_2)\in \cD=\{0,1\}^2$, let
\begin{align}
Y_{j}(d_{1},d_{2}) & =\sfo_{j}(d_{1},d_{2},X)+U_{j,d_{1},d_{2}}.\label{eq:ex_strategic_outcome}
\end{align}
The observed outcome is generated according to 
$Y_{j}=\sum_{(d_{1},d_{2})\in\cD}1\{D_{1}=d_{1},D_{2}=d_{2}\}Y_{j}(d_{1},d_{2}).$
Suppose the individuals are involved in Roy-type decisions: 
\begin{align*}
D_{1} & =1\{Y_{1}(1,D_{2})-Y_{1}(0,D_{2})\ge\sfo_{c1}(Z_{j},X)+U_{c1}\},\\
D_{2} & =1\{Y_{2}(D_{1},1)-Y_{2}(D_{1},0)\ge\sfo_{c2}(Z_{j},X)+U_{c2}\}.
\end{align*}
Namely, each individual chooses 1 if the payoff of choosing 1 over 0 weakly exceeds its cost, given the other individual's action. A key difference from the previous examples is the presence of externalities in the selection process. This selection process is compatible with \eqref{eq:strategic_selection}
with 
\begin{align*}
\sfs_{j}(D_{-j},Z_{j},X) & \equiv\sfo_{j}(1,D_{-j},X)-\sfo_{j}(0,D_{-j},X)-\mu_{cj}(Z_{j},X)\\
V_{j} & \equiv U_{cj}-U_{j,1,D_{-j}}-U_{j,0,D_{-j}}.
\end{align*}
The individuals' social/strategic interaction is captured by the impact of the
other individual's treatment status on player $j$'s payoff, which
corresponds to $\sfs_{j}(1,z_{j},x)-\sfs_{j}(0,z_{j},x)$.

Multiple solutions to the simultaneous equation system \eqref{eq:strategic_selection} may exist, which makes the selection process set-valued \citep{Tamer:2003tr,ciliberto2021market,Balat:2022aa}. For example, suppose the selection process involves strategic substitution, i.e., $\sfs_j(1,z_j,x)-\sfs_{j}(0,z_j,x)\le 0$ for $j=1,2$.\footnote{A similar argument can be applied to games of strategic complementarity and even to models with incoherent predictions.}
The model's prediction is $D\in G(V_1,V_2|Z,X,;\sfs)$ with
\begin{align}
G(v_1,v_2|z,x;\pi)\equiv
\begin{cases}
\{(0, 0)\} & (v_1,v_2)\in S_{\sfs,(0,0)}(z,x)\\
\{(0, 1)\} & (v_1,v_2)\in S_{\sfs,(0,1)}(z,x)\\
\{(1, 0)\} & (v_1,v_2)\in S_{\sfs,(1,0)}(z,x)\\
\{(1, 1)\} & (v_1,v_2)\in S_{\sfs,(1,1)}(z,x)\\
\{(1, 0), (0, 1)\} & (v_1,v_2)\in S_{\sfs,\{(1,0),(0,1)\}}(z,x).
\end{cases}
\label{eq:strategic_sub}
\end{align}
Figure \ref{fig:ex_entry} summarizes the subsets $S_{\sfs,(0,0)}(z,x),\dots, S_{\sfs,\{(1,0),(0,1)\}}(z,x)$.
This model differs from other examples because the selection process itself is incomplete.

We construct a generalized selection equation as in \eqref{eq:gen_sel_eq} by introducing a random variable $V_s:\Omega\to\{0,1\}$ representing an unknown \emph{selection mechanism}. 
Without loss of generality, suppose that the treatment status $D=(1,0)$ ($D=(0,1)$) is selected when $V_s=1$ ($V_s=0$) and multiple values of $D$ are predicted.\footnote{One may represent the selection mechanism by a latent random variable defining a mixture without loss of generality. See \cite{Tamer:2010aa}, \cite{Ponomareva:2011aa}, \citet[Example 2.6]{Molchanov:2018ui} and \cite[][p.377]{Molinari:2020un}.} 
We do not impose any restrictions on the distribution of $V_s$, reflecting the researcher's agnosticism against the selection. We then let $V\equiv(V_1,V_2,V_s)$ be a vector of controls.

Suppoose that $Z$ is independent of  $U\equiv(U_{j,d_1,d_2}, U_{c_j})_{(d_1,d_2)\in \cD,j=1,2}$ given the control variables $(X, V_1, V_2, V_s)$.
 The set-valued control function can be constructed as follows:
\begin{align}
\boldsymbol{V}(D,Z,X;\sfs) & =\begin{cases}
S_{\pi,(0,0)}(Z,X)\times\{0,1\} & \text{if }D=(0,0)\\
\left[S_{\pi,(0,1)}(Z,X)\times\{0,1\}\right]\cup\left[S_{\pi,\{(1,0),(0,1)\}}(Z,X)\times\{0\}\right] & \text{if }D=(0,1)\\
\left[S_{\pi,(1,0)}(Z,X)\times\{0,1\}\right]\cup\left[S_{\pi,\{(1,0),(0,1)\}}(Z,X)\times\{1\}\right] & \text{if }D=(1,0)\\
S_{\pi,(1,1)}(Z,X)\times\{0,1\} & \text{if }D=(1,1).
\end{cases}\label{eq:ex_strategic_svcf}
\end{align}
In words, if $D=(0,0)$ is realized, it implies that $(V_1,V_2)\in S_{\pi,(0,0)}(Z,X)$, regardless of the value of $V_s$. Therefore, $V$ belongs to $S_{\pi,(0,0)}(Z,X)\times\{0,1\}$. Similarly, if $D=(0,1)$ realizes, then either $(0,1)$ is uniquely predicted due to $(V_1,V_2)\in S_{\pi,(0,1)}(Z,X)$ (and $V_s$ unrestricted) or $(0,1)$ is selected from the set of treatment statuses due to $(V_1,V_2)\in S_{\pi,\{(1,0),(0,1)\}}(Z,X)$ and $V_s=0$. The other cases can be analyzed similarly.

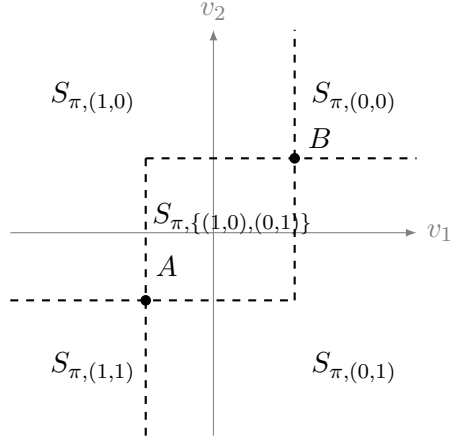
\begin{figure}[h]
	\centering
	\caption{Level sets of $v\mapsto G(v|z;\pi)$ and set-valued CF}
	\begin{tikzpicture}
	[scale=0.9, domain=-3:3,>=latex]
	\draw[gray,->] (-3,0) -- (3,0) node[right] {$v_1$}; 
	\draw[gray,->] (0,-3) -- (0,3) node[above] {$v_2$}; 
	\draw[thick, dashed] (-1,-3) -- (-1,1.1);
	\draw[thick, dashed] (-3,-1) -- (1.2,-1);
	\filldraw (-1,-1) circle (2pt);
	\draw (-1,-0.5) node[right]{$A$} ; 
	\draw[thick, dashed] (1.2,-1) -- (1.2,3);
	\draw[thick, dashed] (-1,1.1) -- (3,1.1);
	\filldraw (1.2,1.1) circle (2pt);
	\draw (1.25,1.4) node[right]{$B$}; 
	\draw (-1,-2) node[left] {$S_{\sfs,(1,1)}$};  
	\draw (1.3,2) node [right]{$S_{\sfs,(0,0)}$};
	\draw (-1,2) node[left] {$S_{\sfs,(1,0)}$}; 
	\draw (1.3,-2) node[right] {$S_{\sfs,(0,1)}$}; 
	\draw (-1.05,0.2) node [right]{$S_{\sfs,\{(1,0),(0,1)\}}$}; 
	\end{tikzpicture}
	\begin{minipage}{0.8\textwidth}
		
	\bigskip
	{\footnotesize Note: $A\equiv(\sfs_{1}(1,z_1,x), \sfs_2(1,z_1,x))$; $B\equiv(\sfs_{1}(0,z_1,x), \sfs_{2}(0,z_2,x))$.  
 The subsets are defined as follows.\begin{align*}
S_{\sfs,(0,0)}(z,x)&\equiv\{v:v_1> \sfs_1(0,z_1,x),v_2>\sfs_2(0,z_2,x)\}\\
S_{\sfs,(0,1)}(z,x)&\equiv\{v:\sfs_1(1,z_1,x)<v_1\le \sfs_1(0,z_1,x),v_2\le \sfs_2(1,z,x)\}\cup \{v:\sfs_1(0,z_1,x)<v_1,v_2\le \sfs_{2}(0,z_2,x)\}\\
S_{\sfs,(1,0)}(z,x)&\equiv\{v:v_1\le \sfs_{1}(1,z_1,x),v_{2}> \sfs_{2}(1,z_2,0)\}\cup \{v: \sfs_{1}(1,z_1,x)<v_1\le \sfs_{1}(0,z_1,x),v_2> \sfs_2(0,z_2,x)\}\\
S_{\sfs,(1,1)}(z,x)&\equiv\{v:v_1\le\sfs_{1}(1,z_1,x),v_2\le \sfs_{2}(1,z_2,x)\}\\
S_{\sfs,\{(0,1),(1,0)\}}(z,x)&\equiv\{v:\sfs_{j}(0,z_j,x)< v_j\le \sfs_{j}(1,z_j,x),j=1,2\}.
\end{align*} \par}
	\end{minipage}
			\label{fig:ex_entry}
\end{figure} 

We work with the conditional mean-independence assumption.
Recall $D\equiv(D_1,D_2)$. Define the model prediction
\begin{align}
	\eY(\eta,D,X,\eV;\sfo,F)\equiv\cl\{y\in\cY:y=\sfo(D,X)+\lambda_D(X,V)+\eta_D,V\in\Sel(\eV)\},
\end{align}
where $\lambda_{d}(x,v)$ is the conditional mean function of $U_{d}|X,V$. Let us rewrite the set-valued control function in \eqref{eq:ex_strategic_svcf} as a union of two random sets.
\begin{align}
\eV(D,X,Z;\sfs)=[\tilde{\eV}_0(D,X,Z;\sfs)\times \{0\}]\cup [\tilde{\eV}_1(D,X,Z;\sfs)\times \{1\}],\label{eq:ex_strategic_svcf_union}
\end{align}
where
\begin{align}
\tilde{\eV}_0(D,X,Z;\sfs)\equiv\begin{cases}
S_{\pi,(0,1)}(Z)\cup S_{\pi,\{(1,0),(0,1)\}}(Z) & \text{if }D=(0,1)\\
S_{\pi,(d_1,d_2)}(Z) & \text{if }D=(d_1,d_2),~ (d_1,d_2)\ne(0,1),
\end{cases}
\end{align}
and
\begin{align}
	\tilde{\eV}_1(D,X,Z;\sfs)\equiv\begin{cases}
S_{\pi,(1,0)}(Z)\cup S_{\pi,\{(1,0),(0,1)\}}(Z) & \text{if }D=(1,0)\\
S_{\pi,(d_1,d_2)}(Z) & \text{if }D=(d_1,d_2),~(d_1,d_2)\ne(1,0).
\end{cases}
\end{align}
As in the previous examples, the sharp identification region for $\theta \equiv(\mu,\pi,F)$ involves the supremum and infimum of a function $f$ over $v\in \eV(d,x,z;\sfs)$. Eq. \eqref{eq:ex_strategic_svcf_union} suggests that the supremum, for example, can be written as 
\begin{align}
\sup_{v\in \eV(d,x,z;\sfs)}f(v_1,v_2,v_s)=\max\{\sup_{(v_1,v_2)\in \tilde{\eV}_0(d,x,z)}f(v_1,v_2,0),\sup_{(v_1,v_2)\in \tilde{\eV}_1(d,x,z)}f(v_1,v_2,1)\}.\label{eq:ex_strategic_max}
\end{align}
We use \eqref{eq:aumann} from Theorem \ref{thm:mean_id} and argue as in Example \ref{ex:Roy1} to characterize the sharp identification region.

\begin{corollary}\label{cor:ex3}
Suppose $E_{P_0}[|Y|]<\infty.$
Suppose $U\equiv(U_{00},U_{10},U_{01},U_{11})$ has a strictly positive conditional density given $(X,V)$.
Suppose, for each $(d_1,d_2)\in\cD$, $E[U_{d_1,d_2}|Z,X,V]=E[U_{d_1,d_2}|X,V],~a.s.$ Then, $\Theta_I(P_0)$ is the set of parameter values $\theta \equiv(\mu,\pi,F)$ such that, for almost all $(d,x)$,
\begin{multline}
 \sup_{z\in\cZ}\Big\{E_{P_0}[Y|D=d,X=x,Z=z]-\lambda_U(d,x,z)\Big\}	\\
 \le \sfo(d,x)\le \\
 \inf_{z\in\cZ}\Big\{E_{P_0}[Y|D=d,X=x,Z=z]-\lambda_L(d,x,z)\Big\},\label{eq:ex_strategic_intersection_bounds}
\end{multline}
where
\begin{align}
	\lambda_U(d,x,z)&\equiv\max\big\{\sup_{(v_1,v_2)\in\tilde\eV_0(d,x,z;\sfs)}{\lambda_d(x,v_1,v_2,0)},\sup_{(v_1,v_2)\in\tilde\eV_1(d,x,z;\sfs)}{\lambda_d(x,v_1,v_2,1)}\big\},\label{eq:ex_strategic1}\\
	\lambda_L(d,x,z)&\equiv\min\big\{\inf_{(v_1,v_2)\in\tilde\eV_0(d,x,z;\sfs)}{\lambda_d(x,v_1,v_2,0)},\inf_{(v_1,v_2)\in\tilde\eV_1(d,x,z;\sfs)}{\lambda_d(x,v_1,v_2,1)}\big\}.\label{eq:ex_strategic2}
\end{align}	
\end{corollary}

\begin{remark}
One may impose further restrictions on the relationship between $U$ and $V_s$ via a priori restrictions on $F$. An example is to assume $U$ is independent (or mean independent) of the selection mechanism conditional on other control variables. This assumption is plausible if $V_s$ is viewed as a signal that is only relevant for the treatment decision (e.g., firms' profitability), but irrelevant for the outcome (e.g., pollution level). Imposing this assumption allows us to exclude $v_s$ from $\lambda_d$: $\lambda_d(x,v_1,v_2,1)=\lambda_d(x,v_1,v_2,0)\equiv\lambda_d(x,v_1,v_2)$. 
\end{remark}

\begin{remark}
With additional assumptions, it is possible to point identify $\sfs$.
Specifically, suppose $(V_{1},V_{2})\perp(Z_{1},Z_{2})|X$, and for each $j$, $Z_j=(Z_{j,k},Z_{j,-k})$ contains a continuous component $Z_{jk}$ supported on $\mathbb R$. Furthermore, $supp(Z_j,X|Z_{-j})=supp(Z_j,X)$, $a.s.$ \cite{Tamer:2003tr} shows, if one can vary the continuous component to push the choice probabilities toward extreme values, i.e., $\pi_j$ is such that $\lim_{z_{j,k}\to-\infty}\sfs_j(0,z_{j,k},z_{j,-k},x)=0$, and $\lim_{z_{j,k}\to\infty}\sfs_j(1,z_{j,k},z_{j,-k},x)=1$, then, $\sfs$ is point identified. This argument, however, requires a variable with a large support (for each player).
\end{remark}

\begin{remark}
The recent work by \cite{ciliberto2021market} simultaneously estimates the parameters of the incomplete entry model and an outcome model (demand and supply). Their approach can be viewed as a generalization of the simultaneous estimation of a triangular system. In contrast, we use the incomplete entry model to construct a set-valued control function, which is a generalization of the control function approach to a triangular system. Moreover, we are interested in identifying treatment effects, while their focus is recovering market primitives.    
\end{remark}

\subsection{Effects of School Assignment under Capacity Constraints}\label{ssec:ex_school}
To illustrate how the framework applies to causal inference on institutional assignment problems with preferences and constraints, we consider the effect of school choice on educational outcomes. Following \cite{bertanha2024causal}, consider a continuum population of students and a set of four schools,
$\mathcal{J}\equiv\{1,\ldots,4\}$ with an exam school (school 4), a magnet school (school 2), other exam or magnet schools (schools 1 and 3). Assume all schools are acceptable to students (i.e., preferred to outside options). Let $Q$ denote the true preference
relation of a student over the set of options $\mathcal{J}$.\footnote{Within this section, we use $Q$ and $P$ to denote preferences, maintaining notational consistency with \cite{bertanha2024causal}. We also use $\boldsymbol{Q}$ for a set-valued control function.} For example, if $Q=(4,2,3,1)$,
then 4 is preferred to 2 (i.e., $4Q2$), 2 is preferred to 3 (i.e., $2Q3$), and so on. 

Let $Y(d)$ be the potential
outcome (e.g., GPA) of a student being assigned to option $d\in\mathcal{J}$. Let $S\equiv\left(S_{1},\ldots,S_{4}\right)\subseteq\mathbb{R}^{4}$
be a vector of placement scores of a student and, for scores $S$ of a student and admission cutoffs
$c\equiv\left(c_{1},\ldots,c_{4}\right)$ of schools, the
set of feasible options of the student is $B(S)\equiv\left\{ j\in\mathcal{J}:S_{j}\geq c_{j}\right\}$.\footnote{For simplicity, assume that placement scores and cutoffs do not have
ties. \cite{bertanha2024causal} consider a more general framework.} Finally, for a given preference $Q$ and a focal school $j\in\mathcal{J}$, define
a \emph{local preference} $Q_{j}\equiv(k,l)\in\mathcal{J}\times\mathcal{J}$
where $k$ is the most preferred school in the set of feasible options plus
$j$ (i.e., $B(S)\cup\{j\}$) and $l$ is the most preferred school
in the set of feasible options minus $j$ (i.e., $B(S)\setminus\{j\}$). Continuing the example with $Q=(4,2,3,1)$, if the student passes all exams except that for school 3 (i.e., $B(S)=\{1,2,4\}$), then the local preference at exam cutoff $c_4$ of focal school 4 is $Q_4=(4,2)$.

Based on students' reported preferences, a centralized assignment mechanism (e.g., a deferred acceptance algorithm) can be employed to match students to schools. If students' reported preferences $P$ are equal to their true preferences
$Q$, then the \emph{local} preference $P_{4}=Q_{4}$ serves as a control
variable, and conditioning on it suffices to apply a regression discontinuity design (RDD). 
Let $Y(l)$ be the counterfactual GPA after being assigned to school $l\in \{2,4\}$. Conditional on the local preference and under the continuity of the distribution of $Y(l)$ ($l\in \{2,4\}$) with respect to the placement score for school 4, $S_4$, one can identify the
effect on GPA of assigning the magnet school (school $4$) relative to the exam school (school $2$) at the cutoff \citep{kirkeboen2016field, Abdulkadiroglu:2019aa}:
\begin{align*}
&E\left[Y(4)-Y(2)\mid Q_{4}=(4,2),S_{4}=c_{4}\right]\\
&=\lim_{\epsilon\downarrow0}E\left[Y\mid P_{4}=(4,2),S_{4}=c_{4}+\epsilon\right]-\lim_{\epsilon\uparrow0}E\left[Y\mid P_{4}=(4,2),S_{4}=c_{4}+\epsilon\right].    
\end{align*}

In reality, school assignment mechanisms face various constraints (e.g., capacity constraints, application costs), in which case, students
have incentives to strategically misreport their preferences \citep{Agarwal:2018aa,Fack:2019aa}. This case is the main focus of \cite{bertanha2024causal}. Suppose
students submit a partial order $P$ of their true preferences
$Q$, subject to a cap of three on the number of schools they can apply to. Then, we do not observe the full preference ordering $Q$, but only a partial order $P$ of $Q$ with $|P|\leq 3$. This introduces ambiguity in using true (local) preferences as controls, which nonetheless can be accommodated by a set-valued control function.

Fix school $4$ with cutoff $c_{4}$.  Consider students who
report preferences over three schools, $P=(4,2,3)$, and have placement scores $S=\left(S_{4},S_{-4}\right)$.
Their best feasible options above and below $c_4$ are $B(S)\cup\{4\}=\{1,2,4\}$ and $B(S)\setminus\{4\}=\{1,2\}$, respectively. Assuming that students are matched to their best feasible options according to $P$, those with $S_4\ge c_4$ are matched to school 4 and those with $S_4 < c_4$ are matched to school 2. In both cases, the reported local preference at $c_4$ is $P_{4}=(4,2)$, which is no longer necessarily equal to true local preference $Q_{4}$. 

Upon observing $P$ and $S$, we can only infer that the true local preference $Q_{4}$ 
belongs to 
\begin{align}\label{local_pref}
    \boldsymbol{Q}_{4}(P,S) & \equiv \begin{cases}
\{(4,2),(4,1)\} & \text{ if }S_{4}\geq c_{4},\\
\{(4,2)\} & \text{ if }S_{4}<c_{4}.
\end{cases}
\end{align}
The intuition is simple: when $S_4\ge c_4$, assignment to school 4 is consistent with multiple underlying preferences (i.e., $Q_4=(4,1)$ or $(4,2)$), whereas when $S_4<c_4$, assignment to school 2 rules out $(4,1)$.

Hence, unlike in the truth-telling benchmark where $P_{4}=Q_{4}=(4,2)$, we can only recover a set of $Q_{4}$'s that is consistent with reported $P=(4,2,3)$. Under the assumption that $Q_{4}\in\boldsymbol{Q}_{4}$ with probability 1 conditional
on $S_{4}$ (cf. Assumption 6(i) in \cite{bertanha2024causal}), the set-valued control $\boldsymbol{Q}_{4}$ can be used to construct sharp bounds on the effect of school assignment, $E\left[Y(4)-Y(2)\mid Q_{4}=(4,2),S_{4}=c_{4}\right]$.

Assume $Y(l)=\mu(l,S_{4})+U$ for $l\in \{2,4\}$ and let $\lambda(Q_{4},S_{4})\equiv E[U|Q_{4},S_{4}]$.\footnote{We can allow for a more general specification $Y(l)=\mu(l,Q_{4},S_{4})+U_l$ for $l\in\{2,4\}$ where $\mu$ is also a function of $Q_4$ and $U_l$ replaces $U$, following Section \ref{sec:mean_indep}. With the vector unobservables, $(U_2,U_4)$, one can define $\lambda_l(Q_{4},S_{4})\equiv E[U_l|Q_{4},S_{4}]$ for $l\in\{2,4\}$. Then, the RDD approach will require assuming $\lambda_4(Q_{4},c_{4})=\lambda_2(Q_{4},c_{4})$, while our partial identification approach do not rely on such an assumption.} Then, the parameter of interest can be expressed as 
$$E\left[Y(4)-Y(2)\mid Q_{4}=(4,2),S_{4}=c_{4}\right]= \mu(4,c_{4})-\mu(2,c_{4}).$$
Conditional on $Q_{4}=(4,2)$ and $S_{4}\geq c_{4}$, we may write 
$$Y=\mu(4,S_{4})+\lambda(Q_{4},S_{4})+\eta,$$
where $\eta\equiv U-\lambda(Q_{4},S_{4})$ satisfies $E[\eta|Q_{4},S_{4}]=0$. Therefore, conditional on $P_{4}=(4,2)$ and $S_{4}\ge c_{4}$, the observed outcome belongs to
\begin{align*}
\boldsymbol{Y}(\eta,P,S;\mu,\lambda) & \equiv\text{cl}\left\{ y\in\mathcal{Y}:y=\mu(4,S_{4})+\lambda(Q_{4},S_{4})+\eta,Q_{4}\in\text{Sel}(\boldsymbol{Q}_{4})\right\}. 
\end{align*}
Note that we can evaluate $\mu(l,S_4)$ at $l=4$ because, according to the definition \eqref{local_pref}, school $4$ is unambiguously determeined as the first element of $Q_4$ when $P_{4}=(4,2)$ and $S_{4}\ge c_{4}$.\footnote{In fact, this is true in the general formulation of $\boldsymbol{Q}_{j}$ in Proposition 2 of \cite{bertanha2024causal}.} That is, school $4$ is the only school truthfully preferred; otherwise, the true preference would not be consistent with the earlier implication that students with $S_4\ge c_4$ are matched to school $4$.

Fix $P_{4}=(4,2)$ and $s_{4}\ge c_{4}$ below. 
Since $Y \in \boldsymbol{Y}(\eta,P,S;\mu,\lambda)$, using the Aumann expectation, we have $E[Y|P,S] \in\mathbb{E}[\boldsymbol{Y}(\eta,P,S;\mu,\lambda)|P,S]$, which yields a moment inequality via support function (with $b=1$):
\begin{align*}
E[Y|P=p,S_{4}=s_4,S_{-4}=s_{-4}]\le
\sup_{q\in\boldsymbol{Q}_{4}(p,s_4,s_{-4})}\{\mu(4,s_4)+\lambda(q,s_4)\},
\end{align*}
where the right hand side uses $E[\sup_{Y\in\text{Sel}(\boldsymbol{Y}(\eta,P,S;\mu,\lambda))}Y|P=p,S_{4}=s_4,S_{-4}=s_{-4}]
=\sup_{q\in\boldsymbol{Q}_{4}(p,s_4,s_{-4})}\{\mu(4,s_4)+\lambda(q,s_4)\}$.

Assume that $\mu(l,\cdot)$ and $\lambda(Q_{4},\cdot)$ are continuous, which implies the usual continuity of $E[Y(l)|Q_4,\cdot]$ for $l\in\{2,4\}$. Taking the limit, we have
\begin{multline}
\lim_{s_4\downarrow c_{4}}E[Y|P=p,S_{4}=s_4,S_{-4}=s_{-4}] \\
\le\lim_{s_4\downarrow c_4} \sup_{q\in\boldsymbol{Q}_{4}(p,s_4,s_{-4})}\{\mu(4,s_4)+\lambda(q,s_4)\}
= \sup_{q\in\boldsymbol{Q}_{4}(p,c_{4},s_{-4})}\{\mu(4,c_{4})+\lambda(q,c_{4})\}.
\end{multline}
A similar argument can be made with $b=-1$ and $s_4\ge c_{4}$, as
well as $b\in\{1,-1\}$ and $s_4<c_{4}$. 

Consequently, we obtain in total four inequalities  for each conditioning value. Note that $\lambda(q,c_{4})$ enters all four inequalities. When $\boldsymbol{Q}_4$ is a singleton, $\lambda(q,c_{4})$ would cancel out when differencing the inequalities to recover the treatment effect. The researcher can further assume separability, such as $\mu(l,S_{4})=\mu_{l} + \tilde\mu(S_{4})$, and derive intersection bounds on $\mu_{l}$. Such a separability assumption would become more plausible once additional covariates $(X,S_{-4})$ are incorporated: $\mu(l,S,X)=\mu_{l}(S,X) + \tilde\mu(S,X)$. 

\begin{remark}
The bounds resulting from our approach are sharp relative to the maintained assumption. They need not coincide with the analytical bounds (that may not be sharp) and the additional sharp bounds derived in \cite{bertanha2024causal}, reflecting differences in both the bounding strategies and the underlying assumptions.
\end{remark}

\section{Empirical Illustration}\label{sec:empirical_illustration}
We illustrate the framework by revisiting \citet{Thornton:2008tp}, who studies whether completing voluntary counseling and testing and learning one’s HIV test result affects health preventive behavior in Malawi. We use the application to show how the framework maps a real empirical design into a structural model and delivers inference for policy-relevant objects.

This exercise is intended as a methodological illustration. The main value of our approach in this setting is that it yields interpretable causal objects even when the linear-IV coefficient is difficult to map to a single local average treatment effect (LATE), and it makes transparent how the empirical conclusions depend on maintained assumptions about selection and treatment response.\footnote{\citet{Thornton:2008tp} studied this problem using a linear IV model with interaction terms and found a modest difference between the effects of learning an HIV-positive diagnosis and those of learning an HIV-negative diagnosis. The study also reported that the TSLS estimate for the HIV-negative group was small.
The conventional interpretation of the TSLS estimand as a LATE should be treated with caution in this setting, because the empirical design involves multiple continuous instruments and continuous covariates. In such environments, the TSLS estimand can be interpreted as a positively weighted average of LATEs only under fairly restrictive assumptions \citep{MogstadTorgovitskyWalters2021aer,blandholetal2022nber,sloczynski2022notinterpretlineariv}.}

The original study was based on the Malawi Diffusion and Ideational Change Project (MDICP), which provided randomized monetary incentives (vouchers) to individuals. These incentives encouraged individuals to visit voluntary counseling and testing (VCT) centers to learn about their HIV status and redeem the vouchers. During a follow-up survey conducted at the respondents’ homes, they were offered to purchase condoms at a discounted price.

The treatment indicator $D_i$ equals one if respondent $i$ completed VCT and learned their HIV test result.
The outcome variable $Y_i$ represents the number of condom purchases. We focus on individuals who were sexually active at the time of the survey and purchased either 0, 3, or 6 condoms, which corresponds to 0, 1, or 2 packs of condoms.\footnote{This subsample accounts for 94\% of the sample of sexually active individuals. Those outside this sample either purchased condoms at a premium (not by packs) or purchased more than 2 packs, which was rare.}
The instrumental variables are $Z_{amt}$, the amount of the monetary voucher, and $Z_{dist}$, the distance to the VCT center, the location of which was randomized.
We used the same set of observable controls as in the original study, a dummy variable for the HIV diagnosis ($X_{HIV}=1$ if positive) and additional controls ($\tilde X$):  a dummy for male, age, simulated average distance to the VCT center, and a district dummy variable. Let $X\equiv(X_{HIV},\tilde X)$.

Consider a simple ordered choice model of condom purchases:\footnote{We treat $Y$ as an ordered discrete variable, which differs from \cite{Thornton:2008tp} who treated $Y$ as a continuous variable.}
\begin{align}
Y=\begin{cases}
0 & \text{if}~ \mu(D,X)+U \le c_{L},\\
3 & \text{if}~ c_L< \mu(D,X)+U \le c_{U},\\
6 & \text{if}~ \mu(D,X)+U > c_U.
\end{cases}	
\end{align}
Here $U$ captures latent demand for prevention. Selection into learning one’s HIV status is described by
\begin{align}
D=1\{\pi(Z,X,V)\ge 0\}. \label{eq:app_selection}
\end{align}
The unobservable $V$ summarizes latent resistance to learning one’s status. Allowing $U$ and $V$ to be dependent accommodates self-selection, since individuals with stronger latent demand for prevention may also be more or less likely to obtain their test results.

We assume that the conditional distribution of $U$ given $V$ can be written as
\begin{align}
U=Q(\eta;X,V)=g(V)+Q(\eta), \label{eq:g_Q}
\end{align}
for a measurable function $g$ and an invertible function $Q$.\footnote{Equivalently, $U\mid V$ belongs to a location family with location parameter $g(V)$.}
Because $Y$ takes only three ordered values, identification reduces to bounding the conditional probabilities of the lowest and highest outcomes, $Y=0$ and $Y=6$. Proposition \ref{prop:empirical} in Appendix \ref{app_sec:application} provides the formal sharp characterization. 

To make inference computationally tractable, we add further structure to functions $\mu$ and $\pi$. The systematic component of condom demand as
\begin{align}
\mu(d,x)=\mu_1 d+\mu_{int}(d\times x_{\text{HIV}})+\tilde{x}'\beta,
\end{align}
where the interaction term allows the effect of learning one’s status to differ between HIV-positive and HIV-negative individuals. 

For treatment selection, our baseline specification is the single-index model
\begin{align}
\pi(z,x,v)=\tilde{\pi}(z,x)-v,
\end{align}
where $\tilde{\pi}(z,x)$ is nonparametrically identified from the propensity score and is estimated by a series Logit estimator. This specification imposes a \emph{common} latent ranking of individuals’ propensity to learn their status across instrument values.

We then relax this common-ranking restriction with a random-coefficient specification,
\begin{align}
\pi(z,x,v)=v_0+v_1 z_{\text{dist}}+z_{\text{amt}}+x'\gamma,
\end{align}
where $v=(v_0,v_1)$ is a realization of $(V_0,V_1)$ and we assume $V_1\le 0$ almost surely.\footnote{The coefficient on $Z_{\text{amt}}$ is assumed to be positive and normalized to one.}
This model allows individuals to trade off voucher amount and travel distance differently, and therefore admits richer selection patterns than the baseline single-index specification. Appendix \ref{app_sec:application} provides further details on how we model the dependence between $U$ and $V$ in each selection model.

We conduct inference for structural objects: the average structural function (ASF), the average treatment effect (ATE), and the \emph{counterfactual switching probability}. Let $Y(d,x_{HIV})$ denote the counterfactual outcome under treatment status $d$ and HIV status $x_{HIV}$. Then the (structural) switching probability is defined as the share of switchers: $P(Y(0,x_{HIV})=0,\;Y(1,x_{HIV})>0).$ 
These structural objects can be expressed as functions of $\theta$.

For each object,
we construct a confidence interval (CI) $CI_n$ using the universal inference method by \cite{KaidoZhang:2024aay}.  For any $\varphi:\Theta\to\mathbb R$, this interval
 covers each value of $\varphi(\theta)$ in its sharp identification region (see $\mathfrak K_I(d)$ in Theorem \ref{thm:functional}) with at least a prescribed confidence level in any finite samples. This method is suitable for our setting because it is designed for functions of $\theta$, robust to incompleteness, and computationally feasible in settings with a moderate number of discrete and continuous covariates. We provide details of the implementation in Appendix  \ref{app_sec:application}.

Tables \ref{tab:baseline} and \ref{tab:swprob} address two policy-relevant questions under the baseline single-index selection model: how learning one’s HIV status changes expected condom demand (the ATE), and what fraction of respondents would switch from buying no condoms to buying some condoms because of the intervention (the counterfactual switching probability). Two patterns emerge. First, the data are only weakly informative about effects for HIV-positive respondents. Second, for the HIV-negative majority, both the sign and the magnitude of the estimated effect depend importantly on the maintained restrictions. Panel A of Table \ref{tab:baseline} illustrates this point. For HIV-negative individuals, the ASF under learning ($d=1$) is tightly bounded around 0.5, whereas the counterfactual ASF without learning ($d=0$) is substantially higher, $[1.327,2.141]$. The resulting ATE is therefore strictly negative, $[-1.598,-0.874]$. Under the baseline model, learning one’s HIV-negative status is thus associated with a reduction in condom demand.

Because HIV-negative individuals constitute the main prevention-relevant group, we next examine how this conclusion changes as additional restrictions are imposed on treatment response and selection. We illustrate that imposing additional identifying restrictions is straightforward within our framework without needing to prove sharpness for each case. We begin with monotone treatment selection (MTS) \citep{manski_pepper00}, which requires that for each $d=0,1$,
\[
E[Y(d)\mid D=1]\ge E[Y(d)\mid D=0].
\]
A natural interpretation is that individuals who choose to learn their HIV status are, holding treatment fixed, more health-conscious and therefore more likely to purchase condoms. Under our specification, this restriction corresponds to $\rho\le 0$. Panel B of Table \ref{tab:baseline} shows that imposing MTS has little effect: the ASF and ATE intervals remain close to those in the baseline specification, suggesting that MTS adds little identifying content in this application.

We next impose monotone treatment response (MTR) \citep{manski1990nonparametric,manski1997monotone} for the HIV-negative group:
\[
Y(1)\ge Y(0)\mid X_{HIV}=0,\quad a.s.
\]
This restriction encodes the view that, all else equal, learning that one is HIV-negative cannot reduce preventive behavior. In our specification, this corresponds to $\mu_1\ge 0$. Panel C of Table \ref{tab:baseline} shows that MTR substantially tightens and shifts the estimates. For HIV-negative individuals, the ATE becomes strictly positive and tightly bounded, $[0.633,0.814]$. This sign reversal is driven by the maintained assumption rather than by the data alone.
Finally, Panel D combines MTS and MTR. Under both restrictions, the HIV-negative ATE remains positive but is much smaller, with CI $[0.03,0.09]$. Taken together, the results from Table \ref{tab:baseline} show that, under the baseline selection model, the effect for HIV-negative individuals is highly assumption-sensitive. Even when monotonicity restrictions force the effect to be positive, the implied increase in condom purchases is modest. By contrast, the effects for HIV-positive individuals remain much less precisely estimated throughout, possibly due to the limited sample size.

Table \ref{tab:swprob} tells a similar story for extensive-margin responses. Without additional restrictions (Panel A), the structural switching probability is bounded above by 0.191 for HIV-positive individuals but is degenerate at zero for HIV-negative individuals, so the baseline model provides no evidence of upward switching for the latter group. Imposing MTS (Panel B) leaves the bounds essentially unchanged, again indicating that MTS contributes little additional identifying power.

Under MTR for HIV-negative individuals (Panel C), the switching probability becomes strictly positive for both groups. For HIV-negative individuals, the interval tightens to $[0.156,0.216]$. This occurs because once learning is restricted from lowering preventive behavior, the model necessarily assigns positive mass to individuals who move from no purchase to some purchase. When MTS and MTR are imposed jointly (Panel D), the HIV-negative interval shrinks to $[0,0.03]$, so zero is no longer ruled out and any extensive-margin response is small at most. For HIV-positive individuals, the bounds remain wide, $[0,0.352]$, and continue to include zero.

Overall, the switching results reinforce the main message from the ASF and ATE estimates: under the baseline selection model, evidence of behavioral change is highly sensitive to the maintained monotonicity restrictions, and for HIV-negative individuals the implied mass of switchers remains limited.

\begin{table}[htbp]
\begin{center}
    \begin{tabular}{lccc}
        \hline
        \noalign{\vskip 2pt}
       &  & HIV+  & HIV$-$ \\[1pt]
        \hline
        \noalign{\vskip 1pt}
      & $d=1$ & $[0.211, 1.538]$ & $[0.482, 0.513]$ \\
Panel A: (Baseline)      
      & $d=0$ & $[0.482, 2.864]$ & $[1.327, 2.141]$ \\
      & ATE   & $[-2.382, 0.693]$ & $[-1.598, -0.874]$ \\[1pt]
\hline
\noalign{\vskip 1pt}
      & $d=1$ & $[0.211, 1.538]$ & $[0.482, 0.513]$ \\
Panel B: (MTS)      
      & $d=0$ & $[0.482, 2.864]$ & $[1.357, 2.111]$ \\
      & ATE   & $[-2.322, 0.754]$ & $[-1.598, -0.814]$ \\[1pt]
\hline
\noalign{\vskip 1pt}
      & $d=1$ & $[1.116, 2.985]$ & $[1.266, 1.327]$ \\
Panel C: (MTR for HIV$-$)      
      & $d=0$ & $[0.121, 1.296]$ & $[0.452, 0.633]$ \\
      & ATE   & $[0.332, 2.744]$ & $[0.633, 0.814]$ \\[1pt]
\hline
\noalign{\vskip 1pt}
      & $d=1$ & $[0.362, 1.779]$ & $[0.754, 0.905]$ \\
Panel D: (MTS \& MTR for HIV$-$)      
      & $d=0$ & $[0.271, 2.291]$ & $[0.754, 0.905]$ \\
      & ATE  & $[-1.417, 1.417]$ & $[0.03, 0.09]$ \\[1pt]
\hline
    \end{tabular}
    \caption{95\% CIs for ASFs and ATE (Baseline Selection Model)}
    \label{tab:baseline}
\end{center}
\end{table}

\begin{table}[htbp]
\begin{center}
    \begin{tabular}{lccc}
        \hline
        \noalign{\vskip 2pt}
       &  & HIV+  & HIV$-$ \\[1pt]
        \hline
        \noalign{\vskip 1pt}
Panel A: (Baseline)      & Switch Pr.   & $[0, 0.191]$ & $[0, 0]$ \\[1pt]
\hline
\noalign{\vskip 1pt}
Panel B: (MTS)           & Switch Pr.   & $[0, 0.196]$ & $[0, 0]$ \\[1pt]
\hline
\noalign{\vskip 1pt}
Panel C: (MTR for HIV$-$)      & Switch Pr.   & $[0.08, 0.618]$ & $[0.156, 0.216]$ \\[1pt]
\hline
\noalign{\vskip 1pt}
Panel D: (MTS \& MTR for HIV$-$)      & Switch Pr.   & $[0, 0.352]$ & $[0, 0.03]$ \\[1pt]
\hline
    \end{tabular}
    \caption{95\% CIs for the Switching Probabilities (Baseline Selection Model)}
    \label{tab:swprob}
\end{center}
\end{table}

\begin{table}[htbp]
\begin{center}
    \begin{tabular}{lccc}
        \hline
        \noalign{\vskip 2pt}
       &  & HIV+  & HIV$-$ \\[1pt]
        \hline
        \noalign{\vskip 1pt}
      & $d=1$ & $[0.452, 2.593]$ & $[0.784, 1.266]$ \\
Panel A: (Baseline)      
      & $d=0$ & $[0.241, 2.382]$ & $[0.814, 1.477]$ \\
      & ATE   & $[-1.719, 2.020]$ & $[-0.573, 0.392]$ \\[1pt]
\hline
\noalign{\vskip 1pt}
      & $d=1$ & $[0.302, 2.683]$ & $[0.543, 1.297]$ \\
Panel B: (MTS)    
      & $d=0$ & $[0.241, 2.352]$ & $[0.783, 1.839]$ \\
      & ATE   & $[-1.719, 2.020]$ & $[-1.297, 0.271]$ \\[1pt]
\hline
\noalign{\vskip 1pt}
      & $d=1$ & $[0.452, 2.472]$ & $[1.025, 1.297]$ \\
Panel C: (MTR for HIV$-$)      
      & $d=0$ & $[0.392, 2.231]$ & $[0.874, 1.236]$ \\
      & ATE   & $[-1.417, 1.900]$ & $[0.030, 0.452]$ \\[1pt]
\hline
\noalign{\vskip 1pt}
      & $d=1$ & $[0.422, 2.502]$ & $[0.905, 1.296]$ \\
Panel D: (MTS \& MTR for HIV$-$)      
      & $d=0$ & $[0.271, 2.295]$ & $[0.754, 1.263]$ \\
      & ATE   & $[-1.598, 2.020]$ & $[0.030, 0.392]$ \\[1pt]
\hline
    \end{tabular}
    \caption{95\% CIs for ASFs and ATE (Random Coefficient Selection Model)}
    \label{tab:rc}
\end{center}
\end{table}

\begin{table}[htbp]
\begin{center}
    \begin{tabular}{lccc}
        \hline
        \noalign{\vskip 2pt}
       &  & HIV+  & HIV$-$ \\[1pt]
        \hline
        \noalign{\vskip 1pt}
Panel A: (Baseline)      & Switch Pr.   & $[0, 0.508]$ & $[0, 0.116]$ \\[1pt]
\hline
\noalign{\vskip 1pt}
Panel B: (MTS)           & Switch Pr.   & $[0, 0.497]$ & $[0, 0.085]$ \\[1pt]
\hline
\noalign{\vskip 1pt}
Panel C: (MTR for HIV$-$)      & Switch Pr.   & $[0, 0.467]$ & $[0, 0.126]$ \\[1pt]
\hline
\noalign{\vskip 1pt}
Panel D: (MTS \& MTR for HIV$-$)      & Switch Pr.   & $[0, 0.497]$ & $[0, 0.111]$ \\[1pt]
\hline
    \end{tabular}
    \caption{95\% CIs for the Switching Probabilities (Random Coefficient Selection Model)}
    \label{tab:swprob_rc}
\end{center}
\end{table}

The results above are obtained under the parsimonious baseline single-index selection model. We now assess how sensitive these conclusions are to relaxing that structure.

Table \ref{tab:rc} reports the corresponding results under the random-coefficient (RC) selection model, which allows individuals to differ in how they trade off voucher amount and distance to the VCT center. The main message is that the negative effect for HIV-negative individuals is not robust to this more flexible specification. In Panel A, the ASF under treatment for HIV-negative individuals is no longer tightly concentrated around 0.5; instead, it ranges from $[0.784,1.266]$, while the ASF under no treatment is of similar magnitude, $[0.814,1.477]$. As a result, the ATE is no longer sign-identified, with CI $[-0.573,0.392]$. This suggests that the sharp negative effect under the baseline model is driven in substantial part by its restrictive common-ranking structure.

This pattern persists across Panels A--D. For HIV-negative individuals, the ATE intervals remain wider than under the baseline model and either include zero or become only modestly positive under the MTR restriction. For HIV-positive individuals, the data remain weakly informative throughout: the ATE intervals are wide in every specification and consistently include zero.

Table \ref{tab:swprob_rc} shows a similar pattern for the switching probability. For HIV-negative individuals, the share of switchers remains small across all specifications, with upper bounds between 0.085 and 0.126 in Panels A--D. Relative to the baseline model, these bounds are looser and, in particular, the strictly positive lower bound obtained under MTR in Panel C disappears. Thus, once the selection model is made more flexible, the data no longer support a robustly positive mass of HIV-negative individuals who switch from buying no condoms to buying some condoms in response to the intervention.

Taken together, the RC results imply that the intervention is unlikely to induce large changes in condom demand among the policy-relevant HIV-negative population despite its nontrivial implementation cost \citep[][Sec. IV-C]{Thornton:2008tp}. More broadly, the comparison between Tables \ref{tab:baseline} and \ref{tab:rc} highlights that conclusions that appear tight under a parsimonious selection model can become substantially weaker once richer forms of selection heterogeneity are allowed.

\section{Concluding remarks}
Observational data are often generated through complex decision processes. Allowing control functions to be set-valued, this paper expands the scope of the control function approach. The proposed framework accommodates, for example, selection processes that involve rich heterogeneity, dynamic optimizing behavior, or social interaction. Our identifying restrictions are inequalities on the conditional choice probabilities. One can conduct inference using moment-based methods or likelihood-based inference methods. Practitioners can use the results of this paper for various purposes. First, they can evaluate social programs nonparametrically, while taking into account potentially complex treatment selection processes. Second, the bounds in our main identification results can easily be combined with a range of shape restrictions and parametric assumptions, allowing practitioners to conduct a sensitivity analysis to assess the additional identifying power of specific assumptions. The tools from random set theory enable us to guarantee sharpness of bounds one obtains in such a sensitivity analysis, without needing to prove sharpness case after case.

\bibliography{HK_references}

\appendix

\section{Comparison with the IV Approach}

\cite{Chesher:2017vu} applies their method to a single-equation IV model and employ the IV assumption. They characterize identified sets for structural parameters, applying random set theory to the level set of unobservables. In this section, we compare our approach with theirs. The main propose of the comparison is to illustrate that the two frameworks are non-nested and complementary.

We summarize the characterization of the identified set in \cite{Chesher:2017vu} with notation close to ours. Let $Y$ be a vector of endogenous variables, $Z$ be a vector of exogenous variables (e.g., IVs), and $U$ be a vector of structural unobservables. Then define a random closed set of $U$ as
\begin{align*}
\boldsymbol{U}(Y,Z;h) & \equiv\{u:h(Y,Z,u)=0\},
\end{align*}
where $h$ is a structural function, the features of which are of interest. Assume $(h,F_{U|Z})\in\mathcal{M}$ where $\mathcal{M}$ incorporates identifying assumptions. \cite{Chesher:2017vu} use the following Artstein's inequality: For $F_{U|Z}(\cdot|z)\in\mathsf F_{U|Z}$
being the distribution of one of the measurable selections of $\boldsymbol{U}(Y,Z;h)$,
\begin{align*}
F_{U|Z}(B|z) & \ge\mathbb{C}_{h}(B|z)
\end{align*}
holds for all closed sets $B\in\mathcal{F}(\mathcal{U})$ where
$\mathbb{C}_{h}(B|Z)\equiv P[\boldsymbol{U}(Y,Z;\mu)\subseteq B|Z]$. Then, the identified set can be characterized as
\begin{align*}
\{(h,F_{U|Z})\in\mathcal{M}:F_{U|Z}(B|Z) & \ge\mathbb{C}_{h}(B|Z),\text{ }a.s.\text{ }\forall B\in \mathcal{F}(\mathcal{U})\}.
\end{align*} They also provide characterization using the Aumman expectation of
$\boldsymbol{U}(Y,Z;h)$.

On the other hand, we characterize the identified set as
\begin{align*}
\{(\mu,F_{U|V}):P_{0}(A|D,Z) & \ge\mathbb{C}_{\mu,F}(A|D,Z),\text{ }a.s.\text{ }\forall A\in\mathcal{F}(\mathcal{Y}),\text{ }\pi\in\Pi_{r}(P_{0})\},
\end{align*}
where $\mathbb{C}_{\mu,F}(A|D,Z)\equiv F_{\eta}(\eY(\eta,D,\eV;\mu,F_{U|V})\subseteq A)$ and $X$ is suppressed. We also provide characterization using the Aumman expectation of $\eY$.

The two approaches share similar features in that only \emph{sets} of unobservables can be recovered from observed data, while stochastic restrictions are imposed on true unobservables. However, the two differ in several ways. First, the CF assumption is imposed on $F_{U|D,V}$ whereas the IV assumption is
imposed on $F_{U|Z}$. Note that the two stochastic assumptions are not nested. 

Second, the CF assumption is imposed in generating $\eY(\eta,D,\eV;\mu,F_{U|V})$, while the IV assumption is imposed when applying Artstein's inequality using $\boldsymbol{U}(Y,Z;h)$. Therefore, even if we were to use the IV approach by treating our $(U,V)$ as their $U$ and $\mu$ as part of $h$, it is unknown how the CF assumption can be utilized in their framework. Third, due to this difference, the containment functional is compared to the observed distribution to construct the identified set in our setting, while it is compared to the unobserved distribution in theirs. This may have implications on implementation in practice. 

In sum, the two frameworks offer complementary tools to practitioners for robust causal analyses. Practitioners can select the most appropriate approach based on the specific model at hand and their belief on the stochastic nature of their problem.

\section{Details on Dynamic Treatment Effects and Additional Examples}\label{sec:additional_examples}
\subsection{Dynamic Treatment Effects}
We complete the characterization of the sharp identification region for the dynamic treatment effect example in Section \ref{ssec:ex_dynamic4}.
Define the following objects
\begin{align*}
H_{Y_2}(y_1,d_1,d_2,x,v;\theta)&\equiv F_{U_2|X,U_1,V_1,V_2}(\sfo_2(y_1,d_1,d_2,x)|x,u_1,v_1,v_2)  \\
H_{D_2}(y_1,d_1,z_2,x,\tilde v;\theta)&\equiv F_{V_2|X,U_1,V_1}(\sfs_2(y_1,d_1,z_2,x)|x,u_1,v_1)\\
H_{Y_1}(d_1,x,v_1;\theta)&\equiv F_{U_1|X,V_1}(\sfo_1(d_1,x)|x,v_1),
\end{align*}
where $\tilde v\equiv(u_1,v_1)$. Recall  $\eV$ was defined as in \eqref{eq:ex_dynamic_eV}. Also define
\begin{align*}
   	\eV_{D_2}(Y_1,D_1,Z,X;\theta)&\equiv\eV_{U_1}(D,X;\sfo_1)\times \eV_{1}(D,Z_1,X;\sfs_1)\\ 
    \eV_{Y_1}(D_1,Z_1,X;\theta)&\equiv\eV_{1}(D,Z_1,X;\sfs_1).
\end{align*}
Then, we can use \eqref{eq:artstein} from Theorem 1 and characterize the sharp identification region as follows. 
\begin{corollary}\label{cor:ex4}
	
Suppose  (i) $U_1\perp Z_1|X,V_1$ (ii) $V_2\perp Z_1|X,U_1,V_1$ and (iii) $U_2\perp (Z_1,Z_2)|X,V_2,U_1,V_1$.  Then, $\Theta_I(P_0)$ is the set of parameter values $\theta \equiv(\sfo_1,\sfo_2,\sfs_1,\sfs_2,F)$ such that, for almost all $(d,x,z)$,
\begin{multline}
\inf_{v\in\eV(d,x,z;\theta)}H_{Y_2}(y_1,d_1,d_2,x,v;\theta)\\\le P_0(Y_2=1|Y_1=y_1,D_1=d_1,D_2=d_2,X=x,Z=z)\\
\le \sup_{v\in\eV(d,x,z;\theta)}H_{Y_2}(y_1,d_1,d_2,x,v;\theta),\label{eq:binary_artstein1}
\end{multline}
\vspace{-1cm}
\begin{multline}
\inf_{\tilde v\in\eV_{D_2}(y_1,d_1,x,z;\theta)}H_{D_2}(y_1,d_1,x,z_2,\tilde v;\theta)\\\le P_0(D_2=1|Y_1=y_1,D_1=d_1,X=x,Z=z)\\
\le \sup_{\tilde v\in\eV_{D_2}(y_1,d_1,x,z;\theta)}H_{D_2}(y_1,d_1,x,z_2,\tilde v;\theta),\label{eq:binary_artstein2}
\end{multline}
and
\begin{multline}
\inf_{v_1\in\eV_{Y_1}(d_1,x,z_1;\theta)}H_{Y_1}(d_1,x, v_1;\theta)\\\le P_0(Y_1=1|D_1=d_1,X=x,Z_1=z_1)\\
\sup_{v_1\in\eV_{Y_1}(d_1,x,z_1;\theta)}H_{Y_1}(d_1,x, v_1;\theta).\label{eq:binary_artstein3}
\end{multline}
\end{corollary}
In the corollary, the conditional independence assumptions (i), (ii), (iii) for the IVs are useful in constructing informative bounds on $\mu_1$, $\pi_2$, and $\mu_2$, respectively. As in Example \ref{ex:Roy1}, $\sfs_1$ can be point identified from the selection equation in period 1 if $V_1\perp Z_1|X$.

\subsection{Multinomial Choice with a Generalized Selection Model}\label{ssec:ex_multinomial}
Suppose an individual chooses an option $Y$ out of mutually exclusive alternatives $\cY=\{1,\dots, J\}$ by maximizing her utility:\footnote{Similar to Section \ref{ssec:roy}, one can allow further heterogeneity by replacing $U_j$ with $U_{j,D}$ in this model.}
\begin{align}
	Y\in \argmax_{j\in \cY}\sfo_j(D,X)+U_j.
\end{align}
 The individual's utility from alternative  $j$ depends on whether she is enrolled in a certain program $(D=1)$ or not $(D=0)$.\footnote{It is also possible to let $\mu$ be a function of individual-specific unobservables (e.g., random coefficients) and treat them as part of $U$. For simplicity, we do not pursue this extension here.}
For example, \cite{Sosa-Rubi:2009aa} analyze the choice of pregnant women in Mexico who choose sites for their obstetric care. The treatment of interest is enrollment in a public health insurance program that provides access to health services for vulnerable populations. 

 Suppose there is a binary instrument (e.g., eligibility), and $D=D(1)Z+D(0)(1-Z)$ with
\begin{align}
	D(z)=1\{\sfs(z,X)\ge V_z\}.
\end{align}
This is the flexible selection model introduced in Example \ref{ex:Roy1}. Allowing for flexibility is relevant in this context, as the insurance program may not be mandatory for the eligible or exclusive against the non-eligible. The set-valued control function is as in \eqref{eq:Roy2_eV}. 

Let $V=(V_0,V_1)$ and let $U_k=Q_k(\eta;X,V),k=1,\dots,J$.
This model's prediction is
 \begin{multline}
 \eY(\eta,D,X,\eV;\sfo,F)\\
 \equiv\Big\{j\in \cY:\sfo_{j}(D,X)\ge \inf_{V\in \Sel(\eV)}\Big(\max_{k\ne j}[\sfo_k(D,X)+Q_k(\eta;X,V)]-Q_{j}(\eta;X,V)\Big)\Big\}.\label{eq:multinom1}
 \end{multline}
Each element of $\eY$ is the maximizer of the utility index $\sfo_k(D,X)+Q_k(\eta;X,V),k=1,\dots,J$ for some  $V\in\Sel(\eV)$. When $\eV$ is singleton-valued, 
\begin{align}
\eY(\eta,D,X,\eV;\sfo,F)=\argmax_{j\in\cY}\sfo_{j}(D,X)+Q_j(\eta;X,V).\label{eq:multinom2}
\end{align}
The model prediction in \eqref{eq:multinom2} nests \citeposs{Petrin:2010aa} specification, which assumes the additive separability of $Q_j$ between functions of $V$ and $\eta$: $Q_j(\eta;X,V)=g(V;\lambda)+Q_j(\eta_j)$.\footnote{In their notation, $g(V;\lambda)$ is $CF(\mu_n;\lambda)$, and $Q_j(\eta_j)$ is $\tilde\varepsilon_{nj}$. Their specification only allows $V$ to shift the location of the conditional distribution of $U$. They show that this specification holds for several parametric models of $U|V$.}

Let $\theta\equiv (\mu_1,\dots,\mu_J,\pi,F)$. Also, let $A\subseteq \cY$. One can show the containment functional is
\begin{multline}
	 \contf(A|D=d,X=x,Z=z)\\
	 =\sum_{\{j_1,\dots,j_m\}\subset A}F_\eta\Big(\sfo_{j_\ell}(d,x)\ge \inf_{v\in \eV(d,x,z;\sfs)}\Big(\max_{k\ne j_\ell}[\sfo_k(d,x)+Q_k(\eta;x,v)]-Q_{j_\ell}(\eta;x,v)\Big),\ell=1,\dots,m\Big).\label{eq:multinom3}
\end{multline}
The containment functional can be computed by simulating $\eta\sim U[0,1]^J.$ The following corollary characterizes $\Theta_I(P_0)$, applying Theorem \ref{thm:identification}.

\begin{corollary}\label{cor:ex2}
Suppose $U\perp Z|X,V$. 
Then, $\Theta_I(P_0)$ is the set of parameter values $\theta \equiv(\mu_1,\dots,\mu_J,\pi,F)$ such that, for almost all $(d,x,z)$,
\begin{multline}
	P_0(Y\in A|D=d,X=x,Z=z)\ge\\
	\sum_{B\subseteq A}F_\eta\Big(\Big\{\sfo_{j_\ell}(d,x)\ge \inf_{v\in \eV(d,x,z;\sfs)}\Big(\max_{k\ne j_\ell}[\sfo_k(d,x)+Q_k(\eta;x,v)]-Q_{j_\ell}(\eta;x,v)\Big)\Big\}\\
\cap \Big\{\sfo_{j_m}(d,x)< \inf_{v\in \eV(d,x,z;\sfs)}\Big(\max_{k\ne j_m}[\sfo_k(d,x)+Q_k(\eta;x,v)]-Q_{j_m}(\eta;x,v)\Big)\Big\}, j_\ell\in B,j_m\in A\setminus B\Big),\\
	~A\subseteq\{1,\dots,J\}.\label{eq:cor2}
\end{multline}
\end{corollary}
As in the previous example, \eqref{eq:cor2} jointly restricts $\mu$, $F$ (through $Q)$, and $\sfs$ (via $\eV$). Suppose further that $V\perp Z|X$. Then, $\sfs$ is point identified as $\sfs(z,x)=P_0(D=1|Z=z,X=x)$, which also ensures $\tilde\sfs$ in \eqref{eq:Roy2_eV} is point identified.

\subsection{Decisions as Corner Solutions or Censored Decisions}
\label{ssec:corner}
Another important class of selection models are models with censored, missing,
or interval data. Consider censored $D$ and suppose
\begin{align}
D & =\max\{D^{*},0\},\label{eq:censored_D}\\
D^{*} & =\pi^{*}(Z,X)+V.\label{eq:censored_D2}
\end{align}
This model is nested in \eqref{eq:gen_sel_eq} as $D=\pi(Z,X,V)\equiv\max\{\pi^{*}(Z,X)+V,0\}$.
If $D$ results as a corner solution of an economic agent's optimization,
we may be interested in $Y=\mu(D,X)+U$. Examples of such $D$ are the
hours of training or the amount of subsidy affecting certain outcomes.
On the other hand, if $D$ results from data censoring, it is reasonable
to consider $Y=\mu(D^{*},X)+U$. The latter $D$ can be viewed as a
special case of interval regressor \citep{Manski:2002um}, namely, $[D^{l},D^{u}]$
that satisfy
\begin{align}
D^{*} & \in[D^{l},D^{u}],\text{ a.s.}\label{eq:interval_D}
\end{align}
This latter case is related to interval data mentioned in Example \ref{ex:no_sel}, although we considered interval variable $D$ rather than interval-observed $V$ here.

Consider \eqref{eq:censored_D}--\eqref{eq:censored_D2}. Then, we
have a set-valued CF as
\begin{align*}
\boldsymbol{V}(d,x,z;\pi^{*}) & =\begin{cases}
[-\pi^{*}(z,x),\infty) & \text{if }d>0,\\
(-\infty,-\pi^{*}(z,x)] & \text{if }d=0.
\end{cases}
\end{align*}
Assume $E[U|Z,X,V]=E[U|X,V]$, which implies $E[U|D,X,V]=E[U|X,V]$. Then, it holds that $E[Y|D,X,V]=\mu(D,X)+\lambda(X,V)$
where $\lambda(x,v)\equiv E[U|X=x,V=v]$. Accordingly, define
\begin{align*}
\eY(\eta,D,X,Z;\mu,F,\pi^{*}) & \equiv\text{cl}\big\{ y\in\cY:y=\sfo(D,X)+\lambda(X,V)+\eta,V\in\Sel(\eV(D,X,Z;\pi^{*}))\big\}.
\end{align*}
Then, by Theorem \ref{thm:mean_id},
\begin{align*}
\mu(d,x)+\lambda_{l}(d,x,z) & \le E[Y|D=d,X=x,Z=z]\le\mu(d,x)+\lambda_{u}(d,x,z),
\end{align*}
where
\begin{align*}
\lambda_{l}(d,x,z) & \equiv\inf_{v\in\boldsymbol{V}(d,x,z;\pi^{*})}\lambda(x,v),\qquad\lambda_{u}(d,x,z)\equiv\sup_{v\in\boldsymbol{V}(d,x,z;\pi^{*})}\lambda(x,v).
\end{align*}

Next, consider \eqref{eq:censored_D2}--\eqref{eq:interval_D}. This
example involves two set-valued control functions:
\begin{align*}
\boldsymbol{V}(D^{l},D^{u},X,Z;\pi^{*}) & \equiv[D^{l}-\pi^{*}(Z,X),D^{u}-\pi^{*}(Z,X)],\\
\boldsymbol{D}^{*}(D^{l},D^{u}) & \equiv[D^{l},D^{u}].
\end{align*}
Note that \eqref{eq:interval_D} implies $V\in\boldsymbol{V}$ and
$D\in\boldsymbol{D}^{*}$ a.s. Assume that $E[U|D^{l},D^{u},X,V]=E[U|X,V]$
and $\mu(d^{*},x)$ is weakly increasing in $d^{*}$.\footnote{Analogous assumptions appear in \cite{Manski:2002um}, although
they do not consider control variables or a treatment selection process.} Then, define
\begin{align*}
\eY(\eta,D^{l},D^{u},X,Z;\mu,F,\pi^{*}) & \equiv\text{cl}\big\{ y\in\cY:y\in\mu(D^{*},X)+\lambda(X,V)+\eta,D^{*}\in\Sel(\boldsymbol{D}^{*}),V\in\Sel(\eV)\big\}.
\end{align*}
Then, by slightly modifying Theorem \ref{thm:mean_id} and its proof, one can show
that
\begin{align*}
\mu(d^{l},x)+\lambda_{l}(d^{l},d^{u},x,z) & \le E[Y|D^{l}=d^{l},D^{u}=d^{u},X=x,Z=z]\le\mu(d^{u},x)+\lambda_{u}(d^{l},d^{u},x,z),
\end{align*}
where
\begin{align*}
\lambda_{l}(d^{l},d^{u},x,z) & \equiv\inf_{v\in\boldsymbol{V}(d^{l},d^{u},x,z)}\lambda(x,v),\qquad\lambda_{u}(d^{l},d^{u},x,z)\equiv\sup_{v\in\boldsymbol{V}(d^{l},d^{u},x,z)}\lambda(x,v).
\end{align*}
\cite{Manski:2002um} focus on characterizing bounds on the conditional
mean $E[Y|D^{*},X]$, which is related to our parameter of interest,
$\mu$. The main difference is that we assume a selection process
and utilize control variables under different identifying assumptions.

\section{Proofs}
\subsection{Proofs of Theorems \ref{thm:identification}, \ref{thm:mean_id}, and \ref{thm:functional}}
\begin{proof}[Proof of Theorem \ref{thm:identification}]
By Assumptions \ref{as:cf_dist}, one may represent the outcome as
$Y=\sfo(D,X,U)=\sfo(D,X,Q(\eta;X,V)).$ By Assumption \ref{as:cf_set}, $V$ is a measurable selection of $\eV$, and therefore $Y$ is a measurable selection of $\eY(\eta,D,X,\eV;\sfo,F)$. 
Therefore, the model's prediction is summarized by 
\begin{align}
	Y\in \eY(\eta,D,X,\eV;\sfo,F),~a.s.
\end{align}
By Assumption \ref{as:cf_set} (ii), $\eV$ is a function of $(D,X,Z)$. Hence, one may condition on $(D,X,\eV)$ by conditioning on $(D,X,Z)$. By Artstein's inequality \citep[see][Theorem A.1.]{Molinari:2020un}, the distribution  $P_0(A|D,X,Z)$ is the conditional law of a measurable selection of $\eY$ if and only if 
\begin{align}
	P_0(A|D,X,Z)\ge \contf(A|D,X,Z),~\forall A\in\cF(\mathbb R^{d_Y}).
\end{align}
This ensures the representation of the sharp identification region by the inequalities above.
\end{proof}

\begin{proof}[Proof of Theorem \ref{thm:mean_id}]
Let $\mathfrak B\equiv \sigma(D,X,Z)$ be the $\sigma$-algebra generated by $(D,X,Z)$. By Assumptions \ref{as:cf_set} and \ref{as:cf_mean}, we may represent the model's set-valued prediction by $\eY$ in \eqref{eq:defeY_mean}, the random set of outcomes $Y=\sfo(D,X)+\lambda_D(X,V)+\eta_D$, where $\eta=(\eta_d,d\in\cD)$ is conditionally mean independent of $D$.
$\eY$ is integrable because its measurable selection $Y$ is assumed to be integrable.
Because of $Y\in \Sel^1(\eY)$, the model's prediction on the conditional mean is summarized by
\begin{align}
	E_{P_0}[Y|\mathfrak B]\in \mathbb E[\eY(\eta,D,X,Z;\sfo,F)|\mathfrak B],~a.s.,\label{eq:mean_id1}
\end{align}
where the right-hand side is the conditional Aumann expectation of $\eY$. Let $b\in\{-1,1\}.$ Then, \eqref{eq:mean_id1} is equivalent to
\begin{align}
	bE_{P_0}[Y|\mathfrak B]\le s(b, \mathbb E[\eY(\eta,D,X,Z;\sfo,F)|\mathfrak B]),~\text{ for all }b\in \{-1,1\}, \label{eq:mean_id2}
\end{align} 
where $s(b,K)\equiv \sup_{x\in K}bx$ is the support function of $K$.

Now, we use the convexification property \citep[][Theorem A.2.]{Molinari:2020un} of the Aumann expectation of random closed sets to exchange the support function and expectation in \eqref{eq:mean_id2}. Technically, this property relies on the probability space used to define $\eY$ \citep[][Sec. 2.1.2]{Molchanov:2017th}. Hence, we proceed as follows.
Let $\Omega=\mathbb R^{d_U}\times  \mathbb R^{d_D}\times\mathbb R^{d_X}\times  \mathbb R^{d_Z}$ be the sample space, and let $\mathfrak F=\mathfrak F_{\mathbb R^{d_U}}\otimes \mathfrak F_{ \mathbb R^{d_D}}\otimes \mathfrak F_{\mathbb R^{d_X}}\otimes \mathfrak F_{\mathbb R^{d_Z}}$ be the product $\sigma$-algebra, where $\mathfrak F_E$ is the Borel $\sigma$-algebra over $E$.
Let $\mathbb F$ be a probability measure on $(\Omega,\mathfrak F)$. Measurable maps $(\eta,D,X,Z)$ are defined on this space.  Consider a measurable rectangle $A=A_\eta\times A_{D,X,Z}$, where $A_\eta\subset\mathbb R^{d_U}$ and $A_{D,X,Z}\subset \mathbb R^{d_D}\times\mathbb R^{d_X}\times  \mathbb R^{d_Z}$. Then, $\mathbb F(A|\mathfrak B)=F_\eta(A_\eta)$.
By Assumption \ref{as:Fuv} and the construction of $\eta$, $F_\eta$ is atomless. Since any $A\in\mathfrak F$ can be approximated by a countable union of measurable rectangles, conclude that $\mathbb F$ is atomless over $\mathfrak B$.\footnote{An event $A'\in\mathfrak B$ is called a $\mathfrak B$-atom if $\mathbb F(0<\mathbb F(A|\mathfrak B)<\mathbb F(A'|\mathfrak B))=0$ for all $A\subset A'$ such that $A\in\mathfrak F$.} 

Therefore, we can apply the convexification theorem \citep[][Theorem A.2.]{Molinari:2020un}, which yields that $\mathbb E[\eY(\eta,D,X,Z;\sfo,F)|\mathfrak B]$ is convex and
\begin{align}
	s(b, \mathbb E[\eY(\eta,D,X,Z;\sfo,F)|\mathfrak B])=E[s(b,\eY(\eta,D,X,Z;\sfo,F))|\mathfrak B],~b\in \{1,-1\}.\label{eq:mean_id3}
\end{align}
For $b=1$,
\begin{align}
E[s(b,\eY(\eta,D,X,Z;\sfo,F))|\mathfrak B]&=E[\sup_{Y\in \Sel(\eY(\eta,D,X,Z;\sfo,F))}Y|\mathfrak B]\notag\\	
&=\mu(d,x)+E[\sup_{V\in \Sel(\eV(D,X,Z;\sfs))}\lambda_d(X,V)|\mathfrak B]\notag\\
&=\mu(d,x)+\sup_{v\in \eV(d,x,z;\sfs)}\lambda_d(x,v).\label{eq:mean_id4}
\end{align}
By \eqref{eq:mean_id2}-\eqref{eq:mean_id4}, $E_{P_0}[Y|D=d,X=x,Z=z]\le\mu(d,x)+\sup_{v\in \eV(d,x,z;\sfs)}\lambda_d(x,v)$. For $b=-1$, the argument is similar.
\end{proof}

\begin{proof}[Proof of Theorem \ref{thm:functional}]
Let $\theta\in\Theta_I(P_0)$. Let $\eW\equiv \{X\}\times \eV$.
Define
\begin{align}
    \mathfrak K_I(d;\theta)\equiv \{\kappa(d)\in\mathbb R:\kappa(d)=\int\varphi(\sfo(d,x,Q(\eta;w))d\eta dF_W(w),~W\in \Sel(\eW)\}.
\end{align}
This set collects the values of $\kappa(d)$ compatible with $\theta$ for some measurable selection $W$ of $\eW$. The sharp identification region for $\kappa(d)$ is
\begin{align}
    \mathfrak K_I(d)=\bigcup_{\theta\in\Theta_I(P_0)} \mathfrak K_I(d;\theta).
\end{align}
Hence, for the conclusion of the theorem, it suffices to show $ \mathfrak K_I(d;\theta)=[\underline{\kappa}(d;\theta), \overline{\kappa}(d;\theta)]$.

For this, we represent $\mathfrak K_I(d;\theta)$ as the Aumann expectation of a random set and apply the convexification theorem. Define
\begin{align}
\eK(d;\theta)\equiv \Big\{r\in \mathbb R:r=\int\varphi(\sfo(d,x,Q(\eta;W))d\eta,~W\in\Sel(\eW)\Big\}.    
\end{align}
Then, by construction, $\mathfrak K_I(d;\theta)$ is the Aumann expectation of $\eK(d;\theta)$. Under the assumption that the underlying probability space is non-atomic, we may apply the convexification theorem \citep[][Theorem A.2.]{Molinari:2020un}. It ensures $\mathfrak K_I(d;\theta)=\mathbb E[\eK(d;\theta)]$ is a convex closed set. Since $\varphi$ is bounded, $\mathfrak K_I(d;\theta)$ is a bounded closed interval.
Again, by Theorem A.2. of \cite{Molinari:2020un}, its upper bound is
\begin{multline}
 s(1,\eK(d;\theta))=s(1,\mathbb E[\eK(d;\theta)])=E[s(1,\eK(d;\theta))]\\
 =E[\sup_{w\in \eW}\int\varphi(\sfo(d,x,Q(\eta;W))d\eta]=E[\sup_{v\in \eV}\int\varphi(\sfo(d,x,Q(\eta;X,v))]=\overline{\kappa}(d;\theta),
\end{multline}
where we used $\eW=\{X\}\times \eV$. The argument for the lower bound is similar and is omitted.
\end{proof}

\subsection{Lemmas}
\begin{lemma}\label{lem:rdset}
	Suppose $\sfo$ is a measurable function. Then,
$\eY(\eta,D,X,\eV;\sfo,F)$ is a random closed set.	
\end{lemma}

\begin{proof}
$\eY(\eta,D,X,\eV;\sfo,F)$ being closed is immediate from the definition. We show its measurability below. Write $\eY(\eta(\omega),D(\omega),X(\omega),\eV(\omega);\sfo,F)$ as $\eY(\omega)$ for short.
Since $\eV$ is a random closed set, there is a sequence $\{V_n\}$ such that $\eV=\text{cl}(\{V_n,n\ge 1\})$ by Theorem 1.3.3 in \cite{Molchanov:2017th}. Let  $\upsilon_n(\omega)=\sfo(D(\omega),X(\omega),Q(\eta(\omega),V_n(\omega)))$ and note that $\eY=\text{cl}(\{\upsilon_n,n\ge 1\})$ by Lemma \ref{lem:castaing}.
Then, for any $x\in\cY$, the distance function
\begin{align}
	\rho(x,\eY(\omega))=\inf\{\|x-y\|,y\in \eY(\omega)\}=\inf\{\|x-\upsilon_n(\omega)\|,n\ge 1\}
\end{align}
is a random variable in $[0,\infty]$. Again, by Theorem 1.3.3 in \cite{Molchanov:2017th}, the conclusion follows.
\end{proof}

Consider a random closed set $\eX$ that is nonempty almost surely. A countable family of selections $\xi_n\in\Sel(\eX),n\ge 1$ is called the \emph{Castaing representation} of $\eX$ if $\eX=\text{cl}(\{\xi_n,n\ge1\})$. Such representation exists for any random closed set \citep{Molchanov:2017th}.
\begin{lemma}\label{lem:castaing}
Let $\eX$ be a random closed set, and let $\{\xi_n,n\ge 1\}$ be its Castaing representation.  For each $\omega\in\Omega,$ let $\eY(\omega)\equiv\text{cl}\{y\in\cY:y=f(\omega,\xi(\omega)),\xi\in\Sel(\eX)\}$ for a measurable map $f:\Omega\times\cX\to\cY$. Then $\eY$ is a random closed set with a Castaing representation $\{\upsilon_n\}$ with $\upsilon_n(\omega)=f(\omega,\xi_n(\omega))$ for $n\ge 1$.
\end{lemma}
\begin{proof}
Le $\{y_n,n\ge 1\}$ be an enumeration of a countable dense set in $\cY$. For each $n,k\ge 1$ and $\omega\in\Omega$, let $C_{k,n}(\omega)=\{x\in \cX:f(\omega,x)\cap B_{2^{-k}}(y_n)\ne\emptyset\}$. Let $\eX_{k,n}\equiv\eX\cap C_{k,n}$ if the intersection is nonempty and let $\eX_{k,n}=\eX$ otherwise. Note that $\eX_{k,n}$ itself is a random closed set. For each $k,n$, there is $m\in\mathbb N$ such that $\xi_m$ is a measurable selection of $\eX_{k,n}$.
For each $\omega$ with $y\in \eY(\omega)$, we have $y\in B_{2^{-k}}(y_n)$ for some $k,n$, and
\begin{align}
	\|y-\upsilon_{k,n}\|\le \|y-y_n\|+\|y_n-\upsilon_{k,n}\|\le 2^{-k+1},
\end{align}
where $\upsilon_{k,n}=f(\omega,\xi_m)$.
Therefore, the conclusion follows.
\end{proof}

\subsection{Proofs of Corollaries}

\begin{proof}[Proof of Corollary \ref{cor:ex1}]
We show Assumptions \ref{as:cf_set}-\ref{as:Fuv} and invoke Theorem \ref{thm:mean_id}.	First, define 
\begin{align}
\eV(D,Z,X;\sfs) & =\begin{cases}
[0,\sfs(Z,X)] & \text{if }D=1\\{}
[\sfs(Z,X),1] & \text{if }D=0.
\end{cases}
\end{align}
Then Assumption \ref{as:cf_set} (i) holds by the selection equation \eqref{eq:roy_sel1}. Assumption \ref{as:cf_set} (ii) holds because $\eV$ is a function of $(D,Z,X)$ and $\sfs$. Assumptions \ref{as:cf_mean}-\ref{as:Fuv} hold by hypothesis. By Theorem \ref{thm:mean_id}, each $\theta \equiv(\sfo,F,\sfs)$ in the sharp identified set satisfies
\begin{align*}
\sfo(d,x)+\lambda_L(d,x,z)\le E_{P_0}[Y|D=d,X=x,Z=z]\le \sfo(d,x)+\lambda_U(d,x,z).
\end{align*}
Rearranging them yields
\begin{align*}
E_{P_0}[Y|D=d,X=x,Z=z]-\lambda_U(d,x,z)\le 	\sfo(d,x)\le  E_{P_0}[Y|D=d,X=x,Z=z]-\lambda_L(d,x,z).
\end{align*}
Note that $\sfo$ does not depend on $z$. Taking the supremum of the lower bounds and taking the infimum of the upper bounds over $z\in\cZ$ yields the desired result.
\end{proof}

\begin{proof}[Proof of Corollary \ref{cor:ex3}]
The main argument is essentially the same as the proof of Corollary \ref{cor:ex1}. Hence, we omit it. Here, we derive \eqref{eq:ex_strategic1}. By Theorem \ref{thm:mean_id},
\begin{align}
	\lambda_U(d,x,z)=\sup_{v\in \eV(d,x,z;\sfs)}\lambda_d(x,v),
\end{align}
where $v=(v_1,v_2,v_s)$. By \eqref{eq:ex_strategic_max}, the identifying restrictions in \eqref{eq:ex_strategic1} follow.  One can show \eqref{eq:ex_strategic2} by a similar argument. 

\end{proof}

\begin{proof}[Proof of Corollary \ref{cor:ex4}]
We first note that, conditional on $(X,U_1,V_1,V_2)$, the endogenous variables $(Y_1,D_1,D_2)$ are a function of the instruments determined by the following triangular system:
\begin{align*}
D_{2} & =1\{\sfs_{2}(Y_{1},D_{1},Z_{2},X)\ge V_{2}\}\\
Y_{1} & =1\{\sfo_{1}(D_{1},X)\ge U_{1}\}\\
D_{1} & =1\{\sfs_{1}(Z_{1},X)\ge V_{1}\}.
\end{align*}
Therefore, Assumption \ref{as:cf_dist}	follows from $U_2\perp(Z_1,Z_2)|X,U_1,V_1,V_2$, which in turn is implied by $(U_1,U_2,V_1,V_2)\perp(Z_1,Z_2)|X$. The sets $\eV_{U_1},\eV_{1},\eV_{2}$ are defined by inverting the triangular system above with respect to $(U_1,V_1,V_2)$, which ensures Assumption \ref{as:cf_set} (i). Assumption \ref{as:cf_set} (ii) also holds because these sets are functions of $(D,Z,X)$ and $\sfs.$
  
  By Theorem \ref{thm:identification}, $\theta\in\Theta_I(P_0)$ if and only if
\begin{align}
	P_0(Y=1|D=d,X=x,Z=z)&\ge  \contf(\{1\}|D=d,X=x,Z=z)\label{eq:proof_ex4_1}\\
	P_0(Y=0|D=d,X=x,Z=z)&\ge  \contf(\{0\}|D=d,X=x,Z=z)\label{eq:proof_ex4_2}.
\end{align}
By \eqref{eq:eYrep} (and as argued in the text),
\begin{align}
 \contf(\{1\}|D=d,X=x,Z=z)&=\inf_{v\in\eV(d,x,z;\sfs)}H(d,x,v)\label{eq:proof_ex4_3}\\
  \contf(\{0\}|D=d,X=x,Z=z)&=1-\sup_{v\in\eV(d,x,z;\sfs)}H(d,x,v)\label{eq:proof_ex4_4}
 \end{align}
The identifying restriction \eqref{eq:binary_artstein1} follows from \eqref{eq:proof_ex4_1}-\eqref{eq:proof_ex4_4} and noting that $P_0(Y=0|D=d,X=x,Z=z)=1-P_0(Y=1|D=d,X=x,Z=z).$ 

The identifying restrictions \eqref{eq:binary_artstein2}-\eqref{eq:binary_artstein3} follow from applying the same argument sequentially. For example, letting $Y\equiv D_2$, $D\equiv (Y_1,D_1)$, $U\equiv V_2$, and $V\equiv (U_1,V_1)$ and applying the argument above yields \eqref{eq:binary_artstein2}.

\end{proof}

\begin{proof}[Proof of Corollary \ref{cor:ex2}]
We show Assumptions \ref{as:cf_dist}-\ref{as:cf_set} and invoke Theorem \ref{thm:identification}.	Assumption \ref{as:cf_dist} holds because $D$ is a function of $(Z,X,V)$ in \eqref{eq:nonmono_sel}, and $U\perp Z|X,V$.   Define 
\begin{align}
\eV(D,Z,X;\sfs)  =\begin{cases}
\left\{ v\in\mathbb R^2:\tilde \sfs(Z,X)+(1-Z)v_{0}+Zv_{1}\ge0\right\}  & \text{if }D=1\\
\left\{ v\in\mathbb R^2:\tilde \sfs(Z,X)+(1-Z)v_{0}+Zv_{1}\le 0\right\}  & \text{if }D=0,
\end{cases}
\end{align}
where $\tilde{\sfs}(z,x)=\sfs(0,x)+z(\sfs(1,x)-\sfs(0,x)).$
 Assumption \ref{as:cf_set} (i) holds by \eqref{eq:Roy2} and \eqref{eq:nonmono_sel}. Also,  Assumption \ref{as:cf_set} (ii) holds because $\eV$ is a function of $(D,Z,X)$ and $\sfs$.
By Theorem \ref{thm:identification}, $\theta$ is in the sharp identified set if and only if $P_0(A|D,X,Z)\ge \contf(A|D,X,Z)$ holds.
 
For the main result, it remains to show \eqref{eq:multinom3}. Let $A\subset \{1,\dots J\}$ and write the model's prediction as $\eY$ for short. Then,
\begin{align*}
	&\{\eY\subseteq A\}= \bigcup_{B\subseteq A} \{\eY=B\}\\
&=	\bigcup_{B\subseteq A}\Big(\bigcap_{j_\ell\in B}\Big\{\sfo_{j_\ell}(D,X)\ge \inf_{V\in \Sel(\eV)}\Big(\max_{k\ne j_\ell}[\sfo_k(D,X)+Q_k(\eta;X,V)]-Q_{j_\ell}(\eta;X,V)\Big)\Big\}\Big)\\
&\qquad \cap\Big(\bigcap_{j_m\in A\setminus B}\Big\{\sfo_{j_n}(D,X)< \inf_{V\in \Sel(\eV)}\Big(\max_{k\ne j_m}[\sfo_k(D,X)+Q_k(\eta;X,V)]-Q_{j_m}(\eta;X,V)\Big)\Big\}\Big).
\end{align*}
 Conditioning on $(D,X,Z)$ and evaluating the probability on the right-hand side by $F_\eta$ yields \eqref{eq:multinom3}.

\end{proof}

\section{Empirical Application}\label{app_sec:application}
Throughout, the structural parameter for this application is $\theta \equiv(\mu, c_L,c_U,g,\sfs,F,F_{\tilde X,V})$. 

\subsection{Target Objects}
Let $X\equiv(X_{HIV},\tilde X)$. We show that the ASF and the structural switching probability can be expressed as functions of the structural parameter $\theta$. We write $\mu(d,x)=\mu(d,x_{HIV},\tilde x)$.

\bigskip
\noindent
\textbf{Average Structural Function:} 
	\begin{multline}
	\varphi(\theta)=\text{ASF}(d,x_{HIV})=\int 3\times [F(c_U-\mu(d,x_{HIV},\tilde x)-g(v))-F(c_L-\mu(d,x_{HIV},\tilde x)-g(v))]\\
	+6\times [1-F(c_U-\mu(d,x_{HIV},\tilde x)-g(v))]dF_{\tilde X,V}(\tilde x,v).
	\end{multline}
Similarly, the average treatment effect $\text{ASF}(1,x_{HIV})-\text{ASF}(0,x_{HIV})$ can be expressed as a function of $\theta$.

\bigskip
\noindent
\textbf{Structural Switching Probability:} 
\begin{multline}
P\big(Y(0,x_{HIV})=0,\;Y(1,x_{HIV})>0\big) \\
= \int \Big[F\big(c_L-\mu(0,x_{HIV},\tilde x)-g(v)\big)
      -F\big(c_L-\mu(1,x_{HIV},\tilde x)-g(v)\big)\Big]_+
      \, dF_{(\tilde X,V)}(\tilde x,v).
\end{multline}

\subsection{Sharp Identifying Restrictions}
Let $F\equiv Q^{-1}$, and let
\begin{align*}
\contf(\{0\}|d,x,z)&=	F(c_L -\mu(d,x)-\sup_{v\in\eV(d,x,z)}g(v))\\
\capf(\{0\}|d,x,z)&=F(c_L-\mu(d,x)- \inf_{v\in\eV(d,x,z)}g(v))\\
\contf(\{6\}|d,x,z)&=1-F( c_U -\mu(d,x)-\inf_{v\in\eV(d,x,z)}g(v))\\
\capf(\{6\}|d,x,z)&=1-F(c_U -\mu(d,x)-\sup_{v\in\eV(d,x,z)}g(v)).
\end{align*}
The following proposition characterizes the sharp identification region. 
\begin{proposition}\label{prop:empirical}
Suppose Assumptions \ref{as:cf_dist}-\ref{as:cf_set} hold. Then,  $\theta \equiv(\sfo,c_L,c_U,F,\sfs)$ is in the sharp identification region $\Theta_I(P_0)$ if and only if
\begin{align}
&\contf(\{0\}|D,X,Z)\le P_0(Y=0|D,X,Z)\le \capf(\{0\}|D,X,Z)\label{eq:app_id1}\\
&\contf(\{6\}|D,X,Z)\le P_0(Y=6|D,X,Z)\le \capf(\{6\}|D,X,Z),~a.s.,\label{eq:app_id2}
\end{align}	
and $\sfs\in \mathsf \Pi_r(P_0)$.
\end{proposition}

\begin{proof}[\rm Proof of Proposition \ref{prop:empirical}]
As before, we write $\eY$ for $\eY(\eta,D,X,\eV;\sfo,F)$.
Let $\varepsilon=Q(\eta)$. Then,
\begin{align*}
0\in \eY &~~ \text{ if } 	\mu(D,X)+\inf_{V\in\Sel(\eV)}g(V) +\varepsilon\le c_L\\
3\in \eY &~~ \text{ if } 	c_L<\mu(D,X)+\sup_{V\in\Sel(\eV)}g(V) +\varepsilon \cap \mu(D,X)+\inf_{V\in\Sel(\eV)}g(V) +\varepsilon \le c_U\\
6\in \eY &~~ \text{ if }    c_U<\mu(D,X)+\sup_{V\in\Sel(\eV)}g(V) +\varepsilon,
\end{align*}
which is equivalent to
\begin{align*}
0\in \eY &~~ \text{ if } 	 \varepsilon\le c_L-\mu(D,X)-\inf_{V\in\Sel(\eV)}g(V)\\
3\in \eY &~~ \text{ if } 	c_L-\mu(D,X)-\sup_{V\in\Sel(\eV)}g(V) <\varepsilon \cap \varepsilon \le c_U- \mu(D,X)-\inf_{V\in\Sel(\eV)}g(V)\\
6\in \eY &~~ \text{ if }    c_U-\mu(D,X)-\sup_{V\in\Sel(\eV)}g(V)<\varepsilon.
\end{align*}
Hence, if  $c_L-\inf_{V\in\Sel(\eV)}g(V)<c_U- \sup_{V\in\Sel(\eV)}g(V)$,
\begin{align*}
\eY&=\begin{cases}
	\{6\}   & c_U -\mu(D,X)-\inf_{V\in\Sel(\eV)}g(V)<\varepsilon\\
	\{3,6\} &c_U-\mu(D,X)- \sup_{V\in\Sel(\eV)}g(V)<\varepsilon \le c_U -\mu(D,X)-\inf_{V\in\Sel(\eV)}g(V)\\
	\{3\}   & c_L-\mu(D,X)-\inf_{V\in\Sel(\eV)}g(V) < \varepsilon \le c_U-\mu(D,X)- \sup_{V\in\Sel(\eV)}g(V)\\
	\{0,3\} & c_L - \mu(D,X)-\sup_{V\in\Sel(\eV)}g(V)  <\varepsilon  \le c_L -\mu(D,X)-\inf_{V\in\Sel(\eV)}g(V)\\
	\{0\}   & \varepsilon \le c_L -\mu(D,X)-\sup_{V\in\Sel(\eV)} g(V).
\end{cases}\end{align*}
If $c_U- \sup_{V\in\Sel(\eV)}g(V)\le c_L-\inf_{V\in\Sel(\eV)}g(V)$,
\begin{align*}
\eY&=\begin{cases}
	\{6\}   & c_U -\mu(D,X)-\inf_{V\in\Sel(\eV)}g(V)<\varepsilon\\
	\{3,6\} &c_L-\mu(D,X)-\inf_{V\in\Sel(\eV)}g(V)<\varepsilon \le c_U -\mu(D,X)-\inf_{V\in\Sel(\eV)}g(V)\\
	\{0,3,6\}   & c_U-\mu(D,X)- \sup_{V\in\Sel(\eV)}g(V) < \varepsilon \le c_L-\mu(D,X)-\inf_{V\in\Sel(\eV)}g(V)\\
	\{0,3\} & c_L - \mu(D,X)-\sup_{V\in\Sel(\eV)}g(V)  <\varepsilon  \le c_U-\mu(D,X)- \sup_{V\in\Sel(\eV)}g(V)\\
	\{0\}   & \varepsilon \le c_L -\mu(D,X)-\sup_{V\in\Sel(\eV)} g(V).
\end{cases}	
\end{align*}
Let $F(\cdot)=Q^{-1}(\cdot)$. Then,
\begin{align}
\contf(\{0\}|d,x,z)&=	F(c_L -\mu(d,x)-\sup_{v\in\eV(d,x,z)} g(v))\label{eq:nu0}\\
\contf(\{3\}|d,x,z)&= [F(c_U-\mu(d,x)- \sup_{v\in\eV(d,x,z)}g(v)) \notag\\
&\qquad- F(c_L-\mu(d,x)-\inf_{v\in\eV(d,x,z)}g(v))]\vee 0\\
\contf(\{6\}|d,x,z)&=1-F( c_U -\mu(d,x)-\inf_{v\in\eV(d,x,z)}g(v))\label{eq:nu6}\\
\contf(\{0,3\}|d,x,z) &=F(c_U -\mu(d,x)-\sup_{v\in\eV(d,x,z)}g(v))\\
\contf(\{3,6\}|d,x,z) &=1-F(c_L-\mu(d,x)- \inf_{v\in\eV(d,x,z)}g(v))\\
\contf(\{0,6\}|d,x,z) &=\contf(\{0\}|x)+\contf(\{6\}|x).\label{eq:red}
\end{align}
Note that
\begin{align}
\capf(\{0\}|x)&=1-\contf(\{3,6\}|x)=F(c_L-\mu(d,x)- \inf_{v\in\eV(d,x,z)}g(v))\label{eq:nustar0}\\
\capf(\{6\}|x)&=1-\contf(\{0,3\}|x)=1-F(c_U -\mu(d,x)-\sup_{v\in\eV(d,x,z)}g(v)).\label{eq:nustar6}
\end{align}
Assumptions \ref{as:cf_dist}-\ref{as:cf_set} ensure that we can invoke Theorem \ref{thm:identification}.
By Theorem \ref{thm:identification}, $\theta$ is in the sharp identification region if and only if $P_0(A|d,x,z)\ge \contf(A|d,x,z),A\in\cC.$
The restriction $P_0(\{0,6\}|d,x,z)\ge \contf(\{0,6\}|d,x,z)$ is redundant because \eqref{eq:red} shows $\contf(\{0,6\}|d,x,z)$ is the sum of $\contf(\{0\}|x)$ and $\contf(\{6\}|x)$.
Note that, for any $A\in \cC$, $P_0(A|d,x,z)\ge \contf(A|d,x,z)$ is equivalent to $P_0(A^c|d,x,z)\le \capf(A^c|d,x,z)$.
Hence, $P_0(\{3\}|d,x,z)\ge \contf(\{3\}|d,x,z)$ is also redundant. This leaves us with the following restrictions
\begin{align*}
P_0(\{0\}|d,x,z)&\ge \contf(\{0\}|d,x,z)\\
P_0(\{6\}|d,x,z)&\ge \contf(\{6\}|d,x,z)\\
P_0(\{0,3\}|d,x,z)&\ge \contf(\{0,3\}|d,x,z)\\
P_0(\{3,6\}|d,x,z)&\ge \contf(\{3,6\}|d,x,z).
\end{align*}
Again,  for any $A\in \cC$, $P_0(A|d,x,z)\ge \contf(A|d,x,z)$ is equivalent to $P_0(A^c|d,x,z)\le \capf(A^c|d,x,z)$. Hence, the restrictions above are equivalent to
\begin{align*}
P_0(\{0\}|d,x,z)&\ge \contf(\{0\}|d,x,z)\\
P_0(\{6\}|d,x,z)&\ge \contf(\{6\}|d,x,z)\\
P_0(\{6\}|d,x,z)&\le \capf(\{6\}|d,x,z)\\
P_0(\{0\}|d,x,z)&\le \capf(\{0\}|d,x,z).
\end{align*}
This completes the proof.
\end{proof}

\subsection{Random Coefficient Selection Model and Parameter Restrictions}\label{ssec:rc_selection_model}
In the random coefficient selection model of Section \ref{sec:empirical_illustration}, let the latent variables have the Gaussian copula representation in 
\begin{align}
U = W_1, 
\quad 
V_0 = W_2, 
\quad 
V_1 = \frac{1}{\lambda}\log(\Phi(W_3)),~\lambda>0,\label{eq:gaussian_copula}
\end{align}
with $(W_1, W_2, W_3) \sim \mathcal{N}(0, \Sigma)$, a trivariate standard normal random variables with correlation matrix $\Sigma$. Below, let $\rho_{ij}=\text{Corr}(W_i,W_j)$.

Note that $W_1 \mid W_2 = w_2, W_3 = w_3 \sim \mathcal{N}(\mu_{1|23}, \sigma^2_{1|23})$ with $\mu_{1|23} = \Sigma_{12:13} \Sigma_{23}^{-1} \begin{bmatrix} w_2 \\ w_3 \end{bmatrix}$ and variance $\sigma^2_{1|23} = 1 - \Sigma_{12:13} \Sigma_{23}^{-1} \Sigma_{12:13}'$ where $\Sigma_{12:13} = (\rho_{12}, \rho_{13})$ and $\Sigma_{23}$ is the correlation matrix of $(W_2,W_3)$. This allows us to represent $U$ as in \eqref{eq:g_Q} with 
\begin{align}
g(v)=
\rho_{12}  v_0 + \rho_{13}  \Phi^{-1}(e^{\lambda v_1}),~~~\text{ and }~~~Q(\eta)=\sqrt{1-\rho^2_{12}-\rho^2_{13}}\Phi^{-1}(\eta).	
\end{align}

One of the additional restrictions considered in Section \ref{sec:empirical_illustration} is the monotone treatment selection (MTS) $E[Y(d)|D=1]\ge E[Y(d)|D=0]$. We impose this assumption as follows. Recall that treated units ($D=1$) are selected from higher $(W_2,W_3)$. Meanwhile, $U=W_1$ shifts the potential outcome vertically. We therefore impose MTS by ensuring that $E[U|W_2,W_3]$ is increasing in both $W_2$ and $W_3$. This is achieved by imposing 
\begin{align}
\rho_{12}-\rho_{23}\rho_{13}\ge0, ~~~\text{and}~~~ \rho_{13}-\rho_{23}\rho_{12}\ge0.
\end{align}

\subsection{Confidence Intervals}
We outline how we construct confidence intervals (CIs) using \cite{KaidoZhang:2024aay}. Toward this end, we write the variables $(D,X,Z)$ as $X$ in order to keep the notation below consistent with the one used in their paper.

Let $\varphi:\Theta\to\mathbb R$, and let $\varphi^*\in\mathbb R$.
Consider the following hypothesis:
\begin{align}
    H_0:\varphi(\theta)=\varphi^*,~~v.s.~~H_1:\varphi(\theta)\ne \varphi^*.
\end{align}
For example, $\varphi(\theta)$ can be the ASF as in the previous section.
A CI for the structural object $\varphi(\theta)$ is obtained by inverting a cross-fit likelihood-ratio (LR) test.  

Their test statistic is constructed as follows. First, split the sample $i=1,\dots,n$ into two subsamples $\mathcal{S}_0$ and $\mathcal{S}_1$ and compute
\begin{align}
	T_n(\varphi^*)\equiv\frac{\cL_0(\hat\theta_1)}{\cL_0(\hat\theta_0)}\equiv\frac{\prod_{i\in \mathcal{S}_0}q_{\hat\theta_1}(Y_i|X_i)}{\sup_{\theta\in \{\theta':\varphi(\theta')=\varphi^*\}}\prod_{i\in \mathcal{S}_0}q_{\theta}(Y_i|X_i)}.\label{eq:def_Tn}
\end{align} 
This is a split-sample likelihood ratio statistic, where a likelihood function (see below) is evaluated at an unrestricted estimator $\hat\theta_1$ and a restricted estimator $\hat\theta_0$.

The unrestricted estimator $\hat\theta_1$ is constructed from sample $\mathcal{S}_1$. We use the minimizer of a sample criterion function $
    \hat{\mathsf Q}_1(\theta)\equiv\sup_{A\in\cC}\sum_{i\in \mathcal{S}_1}\{\contf(A|X_i)-\hat P_1(A|X_i)\}^2_+$ \citep[see, e.g.,][]{chernozhukov2007estimation,Chernozhukov:2013aa} as $\hat\theta_1$, where $\hat P_1$ is the empirical (conditional) distribution of $Y_i$.\footnote{Other choices of $\hat\theta_1$ are also possible as long as they are calculated from $\mathcal{S}_1$.} 
    
The restricted estimator $\hat\theta_0$ is constructed from $\mathcal{S}_0$, 
\begin{align}
    \hat\theta_0\in\argmax_{\theta\in \{\theta':\varphi(\theta')=\varphi^*\}}\prod_{i\in \mathcal{S}_0}q_{\theta}(Y_i|X_i).
\end{align}
 Here, the function $\theta\mapsto q_\theta$ is called the least-favorable-pair (LFP) based density. While we refer to \cite{KaidoZhang:2024aay} for details, we note that this density $q_\theta$ is available in closed form for our application.\footnote{Its derivation is available upon request.}  It represents the worst-case distribution for controlling the test's size when testing a parameter value $\theta$ satisfying the restriction $\varphi(\theta)=\varphi^*$ 
against an unrestricted parameter value $\hat\theta_1$ (treated as a fixed parameter value by conditioning on $\mathcal{S}_1$).

Define the cross-fit LR statistic by \begin{align}
S_n(\varphi^*)\equiv\frac{T_n(\varphi^*)+T_n^{\text{swap}}(\varphi^*)}{2}.	\label{eq:def_Sn}
\end{align}
$T_n^{\text{swap}}(\varphi^*)$ is defined similarly to $T_n(\varphi^*)$ while swapping the roles of $\mathcal{S}_0$ and $\mathcal{S}_1$.
Recall that $\varphi(\theta)\in\mathbb R$ is the target object. 

Define a CI by
\begin{align}
    CI_n\equiv\big\{\varphi^*\in\mathbb R:S_n(\varphi^*)\le \frac{1}{\alpha}\big\}.\label{eq:def_ci}
\end{align}
In our application, we construct a grid of $K=200$ equally spaced points over the range of $\varphi(\theta)$. For each $\varphi^*$ in this grid, we compute $S_n(\varphi^*)$ and compare it to $1/\alpha$, where we use $\alpha=0.05.$ Each CI reported in Section \ref{sec:empirical_illustration} consists of the smallest and largest values of of $\varphi^*$ that are not rejected by the test.

Let $\mathcal P_{\theta,x}\equiv\{Q\in\mathcal M(\Sigma_Y,\cX):Q(A|x)\ge \contf(A|x),~A\in \cC\}$ be the set of conditional probabilities of $Y$ satisfying Artstein's inequality at $\theta$. Let $\cS=\cY\times\cX$, and let $\Delta(\cS^n)$ be the space of probability measures on the product space $\cS^n$. The set of joint distributions of the observable variables compatible with $\theta$ is 
\begin{align}
	\mathcal P^n_\theta\equiv\Big\{P^n\in\Delta(\cS^n):P^n=\bigotimes_{i=1}^n P_i, P_i(\cdot|x)\in\cP_{\theta,x},~	\forall i,~ x\in\cX \Big\}.
\end{align}
Define the sharp identification for $\varphi(\theta)$ by
\begin{align}
\mathcal H_{P^n}[\varphi]\equiv\{\varphi^*:\varphi(\theta)=\varphi^*, P^n\in\cP^n_\theta,\text{ for some }\theta\in\Theta\}.
\end{align}
The CI in \eqref{eq:def_ci} satisfies
\begin{align}
 \inf_{P^n\in \mathcal P^n_\theta,\theta\in\Theta}\inf_{\varphi\in \mathcal H_{P^n}[\varphi]}P^n(\varphi^*\in CI_n)\ge 1-\alpha.
\end{align}
 That is, the CI covers each element of the sharp identification region of $\varphi(\theta)$ with probability at least $1-\alpha$ uniformly over a family of data generating processes compatible with the model. This statement holds in any finite sample, and is called the universal coverage property \citep[][Corollary 1]{KaidoZhang:2024aay}.

\end{document}